\documentclass[a4paper]{article}
\usepackage[top=2.54cm,bottom=2.54cm, left=2.54cm,right=2.54cm]{geometry}

\usepackage{enumitem}
\usepackage{calc}

\usepackage{lmodern} 

\usepackage{graphicx} 

\usepackage{amsmath, accents}
\usepackage{amssymb,tikz}
\usepackage{pict2e}
\usepackage{amsthm}
\usepackage{amscd}
\usepackage{mathrsfs}
\usepackage{todonotes}
\usetikzlibrary{calc} 

\usepackage{yfonts}

\usepackage{marvosym}
\usepackage[percent]{overpic}

\newtheorem{theorem}{Theorem} [section]
	
\newtheorem{corollary}[theorem]{Corollary}	
\newtheorem{lemma}[theorem]{Lemma}

\newtheorem{remark}[theorem]{Remark}
\theoremstyle{definition}

\newcommand{\C}{\mathbb{C}}
\newcommand{\B}{\mathrm{B}}
\newcommand{\D}{\mathbb{D}}

\newcommand{\N}{\mathbb{N}}

\tikzset{
	master/.style={
		execute at end picture={
			\coordinate (lower right) at (current bounding box.south east);
			\coordinate (upper left) at (current bounding box.north west);
		}
	},
	slave/.style={
		execute at end picture={
			\pgfresetboundingbox
			\path (upper left) rectangle (lower right);
		}
	}
}

\let\oldbibliography\thebibliography
\renewcommand{\thebibliography}[1]{\oldbibliography{#1}
\setlength{\itemsep}{-0.5pt}}

\def\XXint#1#2#3{{\setbox0=\hbox{$#1{#2#3}{\int}$}
\vcenter{\hbox{$#2#3$}}\kern-.5\wd0}}

\usepackage{tikz}
\usetikzlibrary{arrows}
\usetikzlibrary{decorations.pathmorphing}
\usetikzlibrary{decorations.markings}
\usetikzlibrary{patterns}
\usetikzlibrary{automata}
\usetikzlibrary{positioning}
\usepackage{tikz-cd}
\tikzset{->-/.style={decoration={
				markings,
				mark=at position #1 with {\arrow{latex}}},postaction={decorate}}}
	
	\tikzset{-<-/.style={decoration={
				markings,
				mark=at position #1 with {\arrowreversed{latex}}},postaction={decorate}}}

\usetikzlibrary{shapes.misc}\tikzset{cross/.style={cross out, draw, 
         minimum size=2*(#1-\pgflinewidth), 
         inner sep=0pt, outer sep=0pt}}

\usepackage{pgfplots}
	
\allowdisplaybreaks	

\numberwithin{equation}{section}

\def\ds{\displaystyle}
\def\bigO{{\cal O}}

\usepackage[colorlinks=true]{hyperref}
\hypersetup{urlcolor=blue, citecolor=red, linkcolor=blue}

\begin{document}
\title{\vspace*{-1.5cm} Eigenvalues of truncated unitary matrices: \\ disk counting statistics}
\author{Yacin Ameur$^*$, Christophe Charlier\footnote{Centre for Mathematical Sciences, Lund University, 22100 Lund, Sweden. e-mails: yacin.ameur@math.lu.se,  christophe.charlier@math.lu.se, joakim.cronvall@math.lu.se
} \, and Philippe Moreillon\footnote{Section of Mathematics, University of Geneva, 1205 Geneva, Switzerland. e-mail: philippe.moreillon@unige.ch}}

\maketitle

\begin{abstract}
Let $T$ be an $n\times n$ truncation of an $(n+\alpha)\times (n+\alpha)$ Haar distributed unitary matrix. We consider the disk counting statistics of the eigenvalues of $T$. We prove that as $n\to + \infty$ with $\alpha$ fixed, the associated moment generating function enjoys asymptotics of the form
\begin{align*}
\exp \big( C_{1} n + C_{2} + o(1) \big),
\end{align*}
where the constants $C_{1}$ and $C_{2}$ are given in terms of the incomplete Gamma function. Our proof uses the uniform asymptotics of the incomplete Beta function. 
\end{abstract}


\section{Introduction}
Let $n \in \N_{>0}$, $\alpha>0$, and consider the joint probability measure
\begin{align}\label{def of point process hard}
\frac{1}{n!Z_{n}}\prod_{1 \leq j<k \leq n}|z_{k}-z_{j}|^{2}  \prod_{j=1}^{n} (1-|z_{j}|^{2})^{\alpha-1} d^{2}z_{j}, \qquad |z_{j}| \leq 1,
\end{align}
where $Z_{n}$ is the normalization constant. Note that the $z_{j}$'s are constrained to lie in the unit disk $\D:=\{z\in \C:|z|\leq 1\}$. A main motivation for studying this point process stems from its connection with random matrices: it is shown in \cite{ZS2000} that for $\alpha \in \mathbb{N}_{>0}$, \eqref{def of point process hard} is the law of the eigenvalues of an $n\times n$ truncated unitary matrix $T$, i.e. $T$ is the upper-left $n\times n$ submatrix of a Haar distributed unitary matrix of size $(n+\alpha)\times (n+\alpha)$. By rewriting \eqref{def of point process hard} in the form
\begin{align*}
\frac{1}{n!Z_{n}}\prod_{1 \leq j<k \leq n}|z_{k}-z_{j}|^{2}  \prod_{j=1}^{n} e^{-nQ(z_{j})} d^{2}z_{j}, \qquad Q(z) := \begin{cases}
-\frac{\alpha-1}{n}\ln(1-|z|^{2}), & \mbox{if } |z|< 1, \\
+\infty, & \mbox{if } |z| \geq 1,
\end{cases}
\end{align*}
we infer that for general $\alpha>0$ (not necessarily $\alpha \in \N_{>0}$), \eqref{def of point process hard} is also the law of a Coulomb gas with $n$ particles at inverse temperature $\beta = 2$ associated with the potential $Q$ \cite{Fo}. 

\medskip We emphasize that \eqref{def of point process hard} is a probability measure only for $\alpha>0$. If $\alpha=0$, the above matrix $T$ is an $n\times n$ Haar distributed unitary matrix, for which the $n$ eigenvalues $\mathrm{z}_{1},\ldots,\mathrm{z}_{n}$ lie exactly on the unit circle according to the probability measure proportional to
\begin{align}\label{CUE density}
\prod_{1 \leq j<k \leq n}|\mathrm{z}_{k}-\mathrm{z}_{j}|^{2}  \prod_{j=1}^{n} d\theta_{j}, \qquad \mathrm{z}_{j} = e^{i\theta_{j}}, \; \theta_{j}\in [0,2\pi).
\end{align}
As is well-known, the equilibrium measure associated with \eqref{CUE density} is the uniform measure on the unit circle.

\medskip This work focuses on the point process \eqref{def of point process hard} as $n\to +\infty$ with $\alpha>0$ fixed. In this regime, $Q(z) = \bigO(n^{-1})$ for any fixed $z\in \D$, and the associated equilibrium measure $\mu$ is defined as the unique measure minimizing the following energy functional
\begin{align*}
\nu \mapsto I[\nu] = \iint \ln \frac{1}{|z-w|}\nu(d^{2}z)\nu(d^{2}w)
\end{align*}
among all Borel probability measures $\nu$ supported on $\D$. This problem is a so-called classical (or unweighted) electrostatics problem, and as such the support of $\mu$ must be the boundary of $\D$ \cite{SaTo} (and this, despite the fact that the point process \eqref{def of point process hard} is two-dimensional). Because the density of \eqref{def of point process hard} is invariant under rotation, we conclude that $\mu$ is the uniform measure on the unit circle $\partial \D$.

\begin{figure}
\begin{center}
\begin{tikzpicture}[master]
\node at (0,0) {\includegraphics[width=4.5cm]{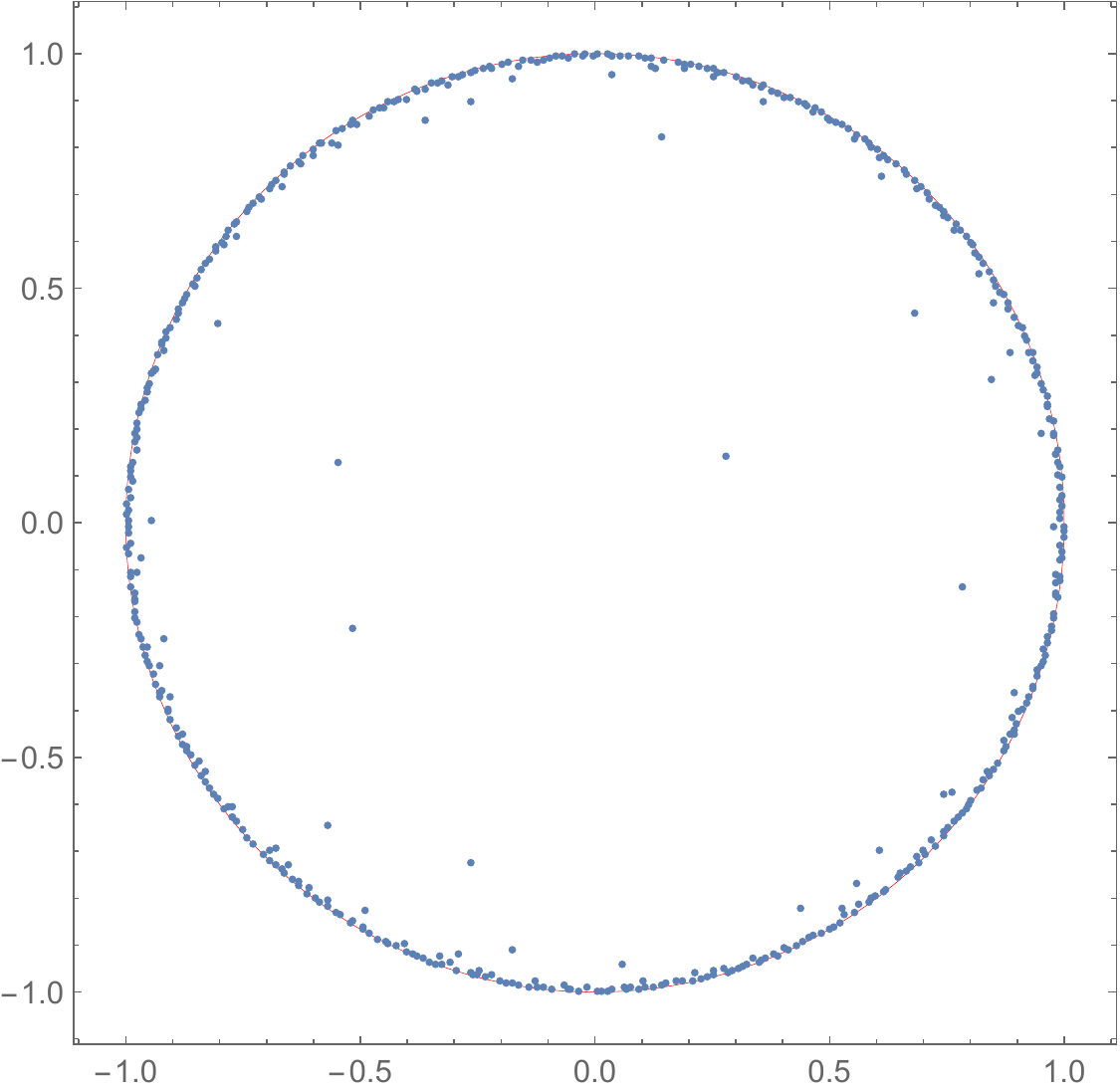}};
\end{tikzpicture}
\begin{tikzpicture}[master]
\node at (0,0) {\includegraphics[width=4.5cm]{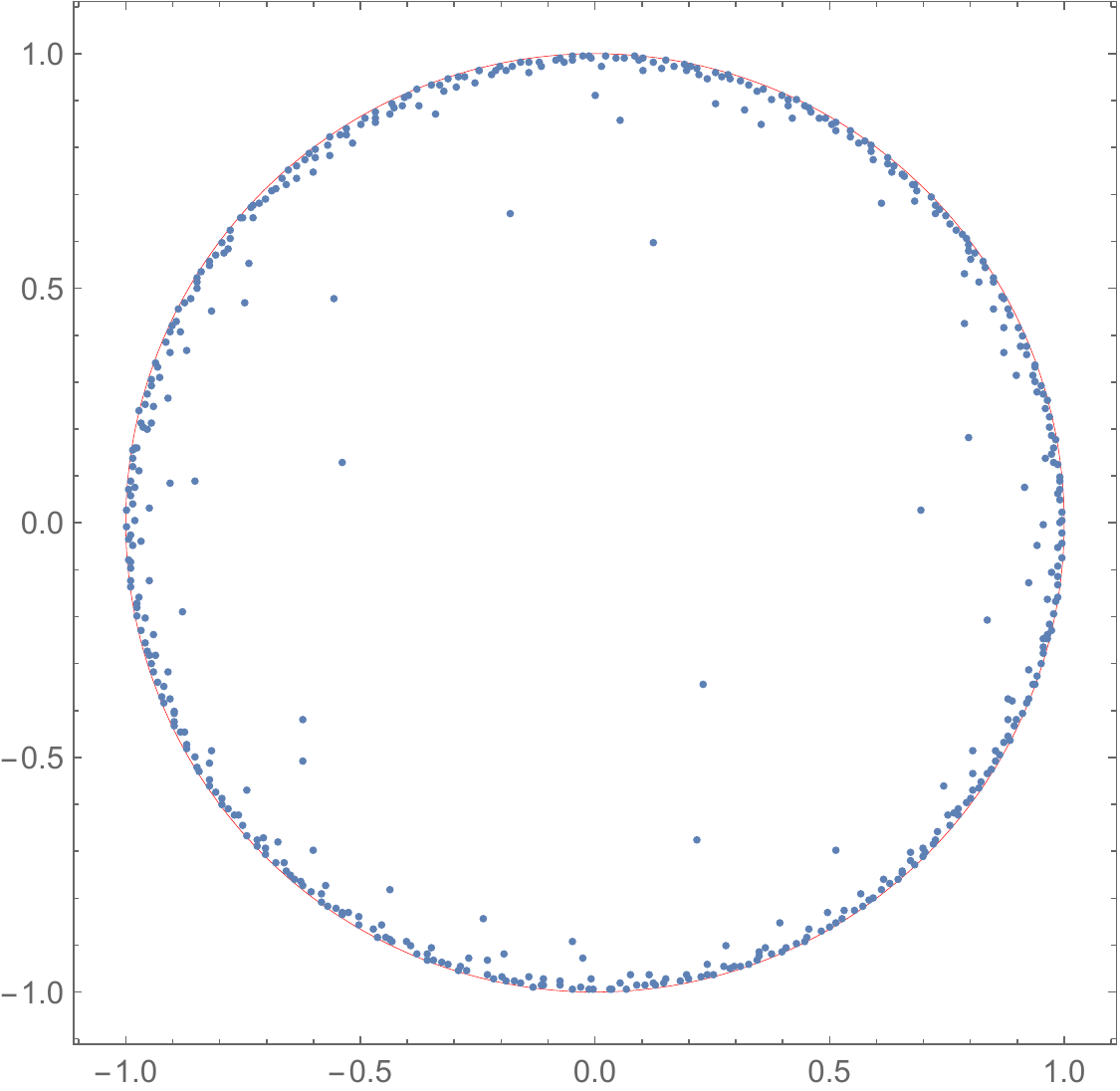}};
\end{tikzpicture}
\begin{tikzpicture}[master]
\node at (0,0) {\includegraphics[width=4.5cm]{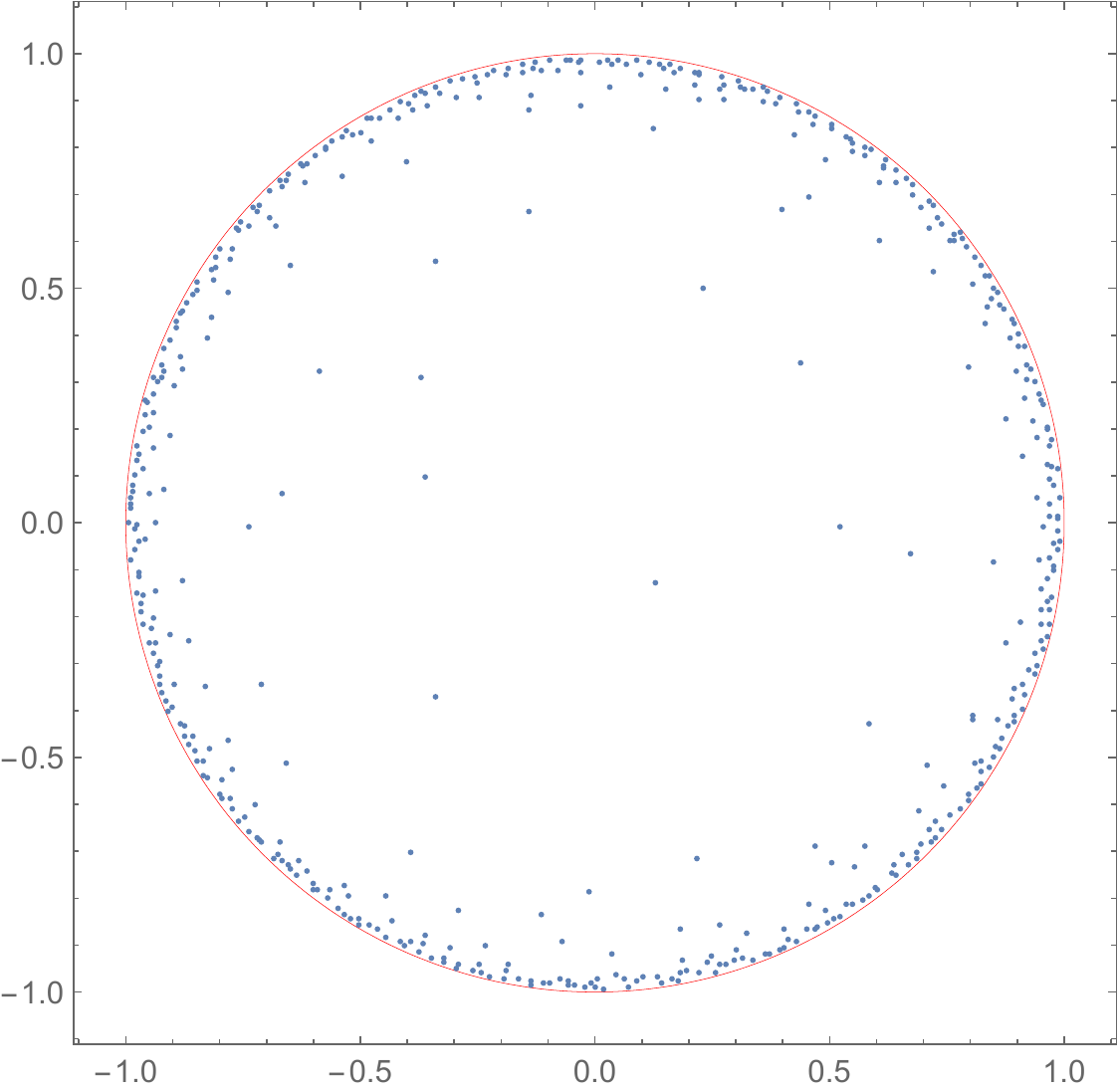}};
\end{tikzpicture}
\end{center}
\caption{\label{fig:TUE} Illustration of the point process \eqref{def of point process hard} with $n=500$ and $\alpha=2$ (left), $\alpha=5$ (middle) and $\alpha=10$ (right). The unit circle is represented in red.}
\end{figure}

\medskip In the language of random matrix theory, $\partial \D$ is a ``hard wall" of \eqref{def of point process hard}. For two-dimensional Coulomb gases, it is a standard fact that along a hard wall the equilibrium measure is singular with respect to the two-dimensional Lebesgue measure; moreover, a non-zero percentage of the points are expected to accumulate (as $n \to + \infty$) in a very small interface of width $1/n$ around the hard wall, see e.g. \cite{SaTo, Seo, ACCL2022}. Following \cite{ACCL2022}, we call this small interface ``the hard edge regime". A particular feature of \eqref{def of point process hard} is that the associated equilibrium measure $\mu$ is purely singular (i.e. $\mu$ has no absolutely continuous component and $\int_{\partial \D}  d\mu =1$). Hence, for large $n$ and $\alpha$ fixed, most of the points $z_{1},\ldots,z_{n}$ of \eqref{def of point process hard} are expected to lie in a $1/n$-neighborhood of $\partial \D$, see also Figure \ref{fig:TUE}. 

\medskip There already exists a fairly rich literature on truncated unitary matrices. For example, the convergence of the distribution of the
maximal modulus $\max_{j}|z_{j}|$ to the Weibull distribution has been studied in \cite{GuiQi2018, LGMS2018, Seo}, several characterizations in terms of Painlevé transcendents for expectations of powers of the characteristic polynomial are established in \cite{DeanoSimm}, and results on the eigenvectors can be found in \cite{Dubach 1}. Also, besides \eqref{def of point process hard}, other two-dimensional point processes whose points are distributed within a narrow interface (or ``band") have been considered, see e.g. \cite{FSK1998, AB2023, BS2021}; however, the point processes considered in these works only feature ``soft edges", and are thus very different from \eqref{def of point process hard}.

\medskip Let $\mathrm{N}(y):=\#\{z_{j}: |z_{j}| < y\}$ be the random variable that counts the number of points of \eqref{def of point process hard} in the disk centered at $0$ of radius $y$. The goal of this paper is to understand the large $n$ behavior of the multivariate moment generating function
\begin{align}\label{moment generating function intro}
\mathbb{E}\bigg[ \prod_{j=1}^{m} e^{u_{j}\mathrm{N}(r_{j})} \bigg]
\end{align}
where $m \in \mathbb{N}_{>0}$ is arbitrary (but fixed), $u_{1},\dots,u_{m} \in \mathbb{R}$ and $r_{1} < \dots <r_{m}$. We consider the hard edge regime, i.e. the radii $r_{1},\dots,r_{m}$ are merging near $1$ at the critical speed $1-r_{j} \asymp n^{-1}$ (we also allow $r_{m}=1$), see also Figure \ref{fig:merging disks}. More precisely, we define
\begin{align}
& r_{\ell} = \Big( 1-\frac{t_{\ell}}{n} \Big)^{1/2}, & & t_{1} > \ldots > t_{m} \geq 0. \label{def of rell part1}
\end{align}
Note that if $t_{m}=0$, then $r_{m}=1$ and trivially $\mathrm{N}(r_{m})=n$ with probability one.
\begin{figure}[h!]
\begin{center}
\begin{tikzpicture}[master]
\node at (0,0) {\includegraphics[width=4.5cm]{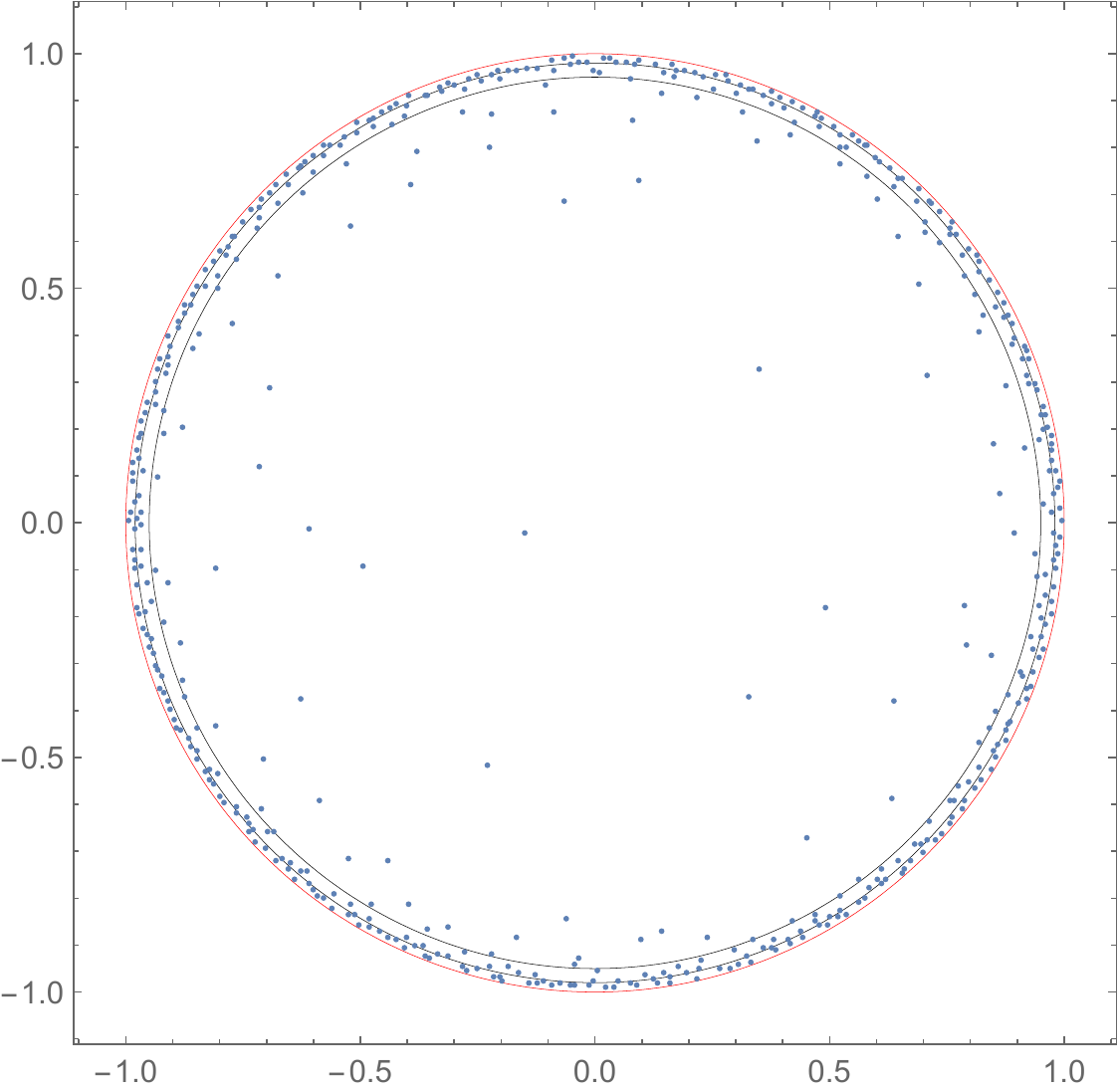}};
\end{tikzpicture}
\begin{tikzpicture}[master]
\node at (0,0) {\includegraphics[width=10cm]{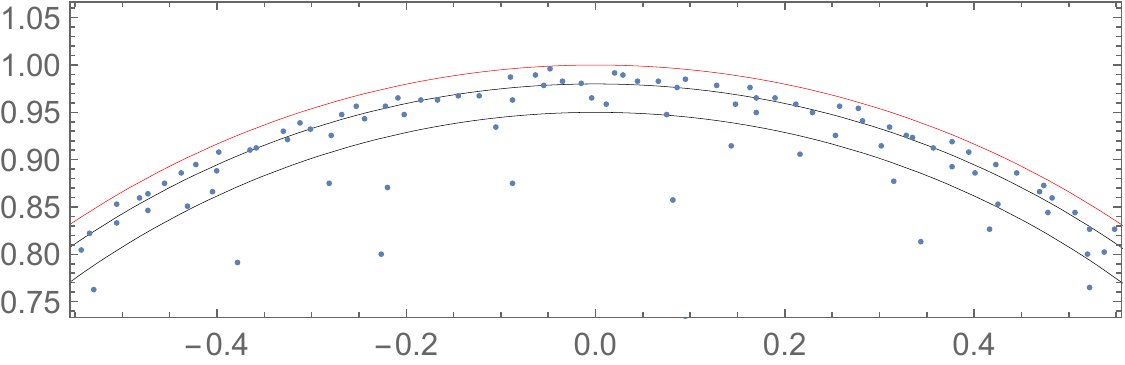}};
\end{tikzpicture}
\end{center}
\caption{\label{fig:merging disks} Left: two circles (in black) merging near the unit circle (in red). Right: a zoom is taken around $i$. For both pictures, $\alpha=10$ and $n=500$.}
\end{figure}

In Theorem \ref{thm:main thm hard} below, we prove that 
\begin{align*}
\mathbb{E}\bigg[ \prod_{j=1}^{m} e^{u_{j}\mathrm{N}(r_{j})} \bigg] = \exp \bigg( C_{1} n + C_{2} + o(1) \bigg), \qquad \mbox{as } n \to + \infty,
\end{align*}
and we give explicit expressions for the constants $C_{1}$ and $C_{2}$ in terms of the following functions
\begin{align}
& \mathcal{H}_{\alpha}(x;\vec{t},\vec{u}) := 1 + \sum_{\ell=1}^{m} (e^{u_{\ell}}-1)\exp \bigg[ \sum_{j=\ell+1}^{m}u_{j} \bigg] \mathrm{Q}(\alpha,t_{\ell}x), \label{def of H} \\
& \mathcal{G}_{\alpha}(x;\vec{t},\vec{u}) := \frac{1}{\mathcal{H}_{\alpha}(x;\vec{t},\vec{u})} \frac{x^{\alpha-1}}{\Gamma(\alpha)} \sum_{\ell=1}^{m} (e^{u_{\ell}}-1)\exp \bigg[ \sum_{j=\ell+1}^{m}u_{j} \bigg] t_{\ell}^{\alpha}e^{-t_{\ell}x}\frac{1-\alpha-t_{\ell}x}{2}, \label{def of G}
\end{align}
where $x\in (0,1]$, $\vec{u}=(u_{1},\ldots,u_{m})\in \C^{m}$, $\vec{t}=(t_{1},\ldots,t_{m})$ is such that $t_{1}>\ldots>t_{m} \geq 0$, $\Gamma(a):=\int_{0}^{+\infty} s^{a-1}e^{-s}ds$ is the Gamma function, and $\mathrm{Q}$ is the normalized incomplete Gamma function:
\begin{align}\label{def of Q}
\mathrm{Q}(a,x) := \frac{\Gamma(a,x)}{\Gamma(a)}, \qquad \Gamma(a,x) := \int_{x}^{+\infty} s^{a-1}e^{-s}ds.
\end{align}
Note that $\mathcal{H}_{\alpha}$ is also well-defined at $x=0$, while $\mathcal{G}_{\alpha}$ is well-defined at $x=0$ only for $\alpha \geq 1$. 

\medskip The statement of our main theorem involves $\ln \mathcal{H}_{\alpha}$, and the function $\mathcal{H}_{\alpha}$ also appears in the denominator of \eqref{def of G}. The following lemma implies that $\ln \mathcal{H}_{\alpha}$ and $\mathcal{G}_{\alpha}$ are well-defined and real-valued for $x \in (0,1]$, $\vec{u}=(u_{1},\ldots,u_{m}) \in \mathbb{R}^{m}$, and $t_{1}>\ldots>t_{m} \geq 0$. In this paper, $\ln$ always denotes the principal branch of the logarithm. 
\begin{lemma}\label{Hjpositivelemma}
$\mathcal{H}_{\alpha}(x;\vec{t},\vec{u})>0$ for all $x \in (0,1]$, $\vec{u}=(u_{1},\ldots,u_{m}) \in \mathbb{R}^{m}$, $t_{1}>\ldots>t_{m} \geq 0$. 
\end{lemma}
\begin{proof}
Since
\begin{align*}
\partial_{u_{1}}\mathcal{H}_{\alpha}(x;\vec{t},\vec{u}) = e^{u_{1}+\ldots+u_{m}} \mathrm{Q}(\alpha,t_{1}x)>0,
\end{align*}
it only remains to verify that $\mathcal{H}_{\alpha}|_{u_{1}=-\infty}\geq 0$. Setting $u_{1}=-\infty$ in \eqref{def of H} and rearranging the terms, we find
\begin{align*}
\mathcal{H}_{\alpha}(x;\vec{t},\vec{u})|_{u_{1}=-\infty} = \big[ 1-\mathrm{Q}(\alpha,t_{m}x) \big] + \sum_{\ell=2}^{m} e^{u_{\ell}+\ldots+u_{m}}[\mathrm{Q}(\alpha,t_{\ell}x)-\mathrm{Q}(\alpha,t_{\ell-1}x)].
\end{align*}
Recall that $t_{1}>\ldots>t_{m} \geq 0$ and $x\in (0,1]$. Hence, since $\mathbb{R}\ni s \mapsto \mathrm{Q}(\alpha,s)$ decreases from $1$ to $0$, the $m$ terms in the above right-hand side are all $\geq 0$, which proves $\mathcal{H}_{\alpha}|_{u_{1}=-\infty}\geq 0$.
\end{proof}
\begin{theorem}\label{thm:main thm hard}
Let $m \in \mathbb{N}_{>0}$, $\alpha>0$ and $t_{1}>\dots>t_{m} \geq 0$ be fixed parameters. For $n \in \mathbb{N}_{>0}$, define
\begin{align}\label{rellhardedge}
& r_{\ell} = \Big( 1-\frac{t_{\ell}}{n} \Big)^{1/2}, \qquad \ell=1,\dots,m.
\end{align}
For any fixed $x_{1},\dots,x_{m} \in \mathbb{R}$, there exists $\delta > 0$ such that 
\begin{align}\label{asymp in main thm hard}
\mathbb{E}\bigg[ \prod_{j=1}^{m} e^{u_{j}\mathrm{N}(r_{j})} \bigg] = \exp \bigg( C_{1} n + C_{2} + \bigO(n^{-\frac{2\hat{\alpha}+\alpha}{2+\alpha}}) \bigg), \qquad \mbox{as } n \to + \infty
\end{align}
uniformly for $u_{1} \in \{z \in \mathbb{C}: |z-x_{1}|\leq \delta\},\dots,u_{m} \in \{z \in \mathbb{C}: |z-x_{m}|\leq \delta\}$, where $\hat{\alpha} := \min\{\alpha,1\}$ and
\begin{align}\nonumber
& C_{1} = \int_{0}^{1}\ln \mathcal{H}_{\alpha}(x;\vec{t},\vec{u}) \, dx, \\
& C_{2} = \int_{0}^{1} \mathcal{G}_{\alpha}(x;\vec{t},\vec{u}) \, dx + \frac{\ln \mathcal{H}_{\alpha}(1;\vec{t},\vec{u}) - \sum_{j=1}^{m}u_{j}}{2}.
\end{align}
In particular, since $\mathbb{E}\big[ \prod_{j=1}^{m} e^{u_{j}\mathrm{N}(r_{j})} \big]$ is analytic in $u_{1},\dots,u_{m} \in \mathbb{C}$ and is positive for $u_{1},\dots,u_{m} \in \mathbb{R}$, the asymptotic formula \eqref{asymp in main thm hard} combined with Cauchy's formula implies that
\begin{align}\label{der of main result hard}
\partial_{u_{1}}^{k_{1}}\dots \partial_{u_{m}}^{k_{m}} \bigg\{ \ln \mathbb{E}\bigg[ \prod_{j=1}^{m} e^{u_{j}\mathrm{N}(r_{j})} \bigg] - \Big( C_{1} n + C_{2} \Big) \bigg\} = \bigO(n^{-\frac{2\hat{\alpha}+\alpha}{2+\alpha}}), \qquad \mbox{as } n \to + \infty,
\end{align}
for any $k_{1},\dots,k_{m}\in \mathbb{N}$, and $u_{1},\dots,u_{m}\in \mathbb{R}$. 
\end{theorem}

Let $(\mathbb{N}^{m})_{>0} := \{\vec{j}=(j_{1},\dots,j_{m}) \in \mathbb{N}^{m}: j_{1}+\dots+j_{m}\geq 1\}$. For $\vec{j} \in (\mathbb{N}^{m})_{>0}$, the joint cumulant $\kappa_{\vec{j}}=\kappa_{\vec{j}}(r_{1},\dots,r_{m};n,\alpha)$ of $\mathrm{N}(r_{1}), \dots, \mathrm{N}(r_{m})$  is defined by
\begin{align}\label{joint cumulant}
\kappa_{\vec{j}}=\kappa_{j_{1},\dots,j_{m}}:=\partial_{\vec{u}}^{\vec{j}} \ln \mathbb{E}[e^{u_{1}\mathrm{N}(r_{1})+\dots + u_{m}\mathrm{N}(r_{m})}] \Big|_{\vec{u}=\vec{0}},
\end{align}
where $\partial_{\vec{u}}^{\vec{j}}:=\partial_{u_{1}}^{j_{1}}\dots \partial_{u_{m}}^{j_{m}}$ and $\vec{0}:=(0,\ldots,0)$. For instance, we have
\begin{align*}
\mathbb{E}[\mathrm{N}(r)] = \kappa_{1}(r), \qquad \mbox{Var}[\mathrm{N}(r)] = \kappa_{2}(r) =\kappa_{(1,1)}(r,r), \qquad \mbox{Cov}[\mathrm{N}(r_{1}),\mathrm{N}(r_{2})] = \kappa_{(1,1)}(r_{1},r_{2}).
\end{align*}
\begin{corollary}\label{coro:correlation hard}
Let $m \in \mathbb{N}_{>0}$, $\vec{j} \in (\mathbb{N}^{m})_{>0}$, $\alpha > 0$,  and $t_{1}>\dots>t_{m} \geq 0$ be fixed. For $n \in \mathbb{N}_{>0}$, define $\{r_\ell\}_{\ell =1}^m$ by \eqref{rellhardedge}.

(a) The joint cumulant $\kappa_{\vec{j}}$ satisfies
\begin{align}\label{asymp cumulant hard edge}
\kappa_{\vec{j}} = \partial_{\vec{u}}^{\vec{j}}C_{1}\big|_{\vec{u}=\vec{0}} \; n + \partial_{\vec{u}}^{\vec{j}}C_{2}\big|_{\vec{u}=\vec{0}} + \bigO\big(n^{-\frac{2\hat{\alpha}+\alpha}{2+\alpha}}\big), \qquad  n \to +\infty,
\end{align}
where $C_{1},C_{2}$ are as in Theorem \ref{thm:main thm hard} and $\hat{\alpha} := \min\{\alpha,1\}$. In particular, for any $1 \leq \ell < k \leq m$,
\begin{align*}
& \mathbb{E}[\mathrm{N}(r_{\ell})] = b_1(t_\ell) n + c_1(t_\ell) + \bigO\big(n^{-\frac{2\hat{\alpha}+\alpha}{2+\alpha}}\big),
 	\\
& \mathrm{Var}[\mathrm{N}(r_{\ell})] = b_{(1,1)}(t_{\ell},t_{\ell})n + c_{(1,1)}(t_{\ell},t_{\ell}) + \bigO\big(n^{-\frac{2\hat{\alpha}+\alpha}{2+\alpha}}\big),
	 \\
& \mathrm{Cov}(\mathrm{N}(r_{\ell}),\mathrm{N}(r_{k})) = b_{(1,1)}(t_{\ell},t_{k})n + c_{(1,1)}(t_{\ell},t_{k}) + \bigO\big(n^{-\frac{2\hat{\alpha}+\alpha}{2+\alpha}}\big) \nonumber
\end{align*}
as $n \to + \infty$, where
\begin{align}
 b_1(t_\ell) =&\; \int_{0}^{1} \mathrm{Q}(\alpha,t_{\ell}x)dx = \mathrm{Q}(\alpha,t_{\ell}) + \alpha \frac{1-\mathrm{Q}(\alpha+1,t_{\ell})}{t_{\ell}}, \label{def of b1} \\ 
c_1(t_\ell) = & \frac{t_{\ell}^{\alpha}}{\Gamma(\alpha)} \int_{0}^{1} x^{\alpha-1}e^{-t_{\ell}x}\frac{1-\alpha-t_{\ell}x}{2}dx + \frac{\mathrm{Q}(\alpha,t_{\ell})-1}{2}, \nonumber
\end{align}
and, for $\ell \leq k$,
\begin{align}\label{def of b11 hard edge}
b_{(1,1)}(t_{\ell},t_{k}) = &\; \int_{0}^{1} \mathrm{Q}(\alpha,t_{\ell}x)\big(1-\mathrm{Q}(\alpha,t_{k}x)\big)dx,
	 \\ \nonumber
c_{(1,1)}(t_{\ell},t_{k}) = &\; \int_{0}^{1}  \bigg\{ (\alpha-1+t_{k}x)t_{k}^{\alpha}e^{-t_{k}x}\mathrm{Q}(\alpha,t_{\ell}x) - (\alpha-1+t_{\ell}x)t_{\ell}^{\alpha}e^{-t_{\ell}x} \big(1-\mathrm{Q}(\alpha,t_{k}x)\big) \bigg\}\frac{x^{\alpha-1}dx}{2\, \Gamma(\alpha)} \\
& + \frac{1}{2}\mathrm{Q}(\alpha,t_{\ell})\big(1-\mathrm{Q}(\alpha,t_{k})\big).
\end{align}

(b) Assume furthermore that $t_{m}>0$. As $n \to + \infty$, the random variable $(\mathcal{N}_{1},\dots,\mathcal{N}_{m})$, where
\begin{align}
& \mathcal{N}_{\ell} := \frac{\mathrm{N}(r_{\ell})-b_1(t_\ell) n}{\sqrt{b_{(1,1)}(t_\ell,t_\ell) n}}, \qquad \ell=1,\dots,m,
\label{Nj hard edge}
\end{align}
convergences in distribution to a multivariate normal random variable of mean $(0,\dots,0)$ whose covariance matrix $\Sigma$ is given by
\begin{align*}
\Sigma_{\ell,k} =  \Sigma_{k, \ell} = \frac{b_{(1,1)}(t_{\ell},t_{k})}{\sqrt{b_{(1,1)}(t_{\ell},t_{\ell})b_{(1,1)}(t_{k},t_{k})}}, \qquad 1 \leq \ell \leq k \leq m.
\end{align*}
\end{corollary}
\begin{remark}\label{remark:N=n with prob 1}
If $t_{m}=0$, then $b_{1}(t_{m})=n$ and $c_{1}(t_{m})=b_{(1,1)}(t_{m},t_{m})=c_{(1,1)}(t_{m},t_{m})=0$, which is consistent with the fact that $\mathrm{N}(r_{m})=n$ with probability $1$. This is the reason why we required $t_{m}>0$ in Corollary \ref{coro:correlation hard} (b).
\end{remark}
\begin{figure}[h]
\begin{center}
\begin{tikzpicture}[master]
\node at (0,0) {\includegraphics[width=6.5cm]{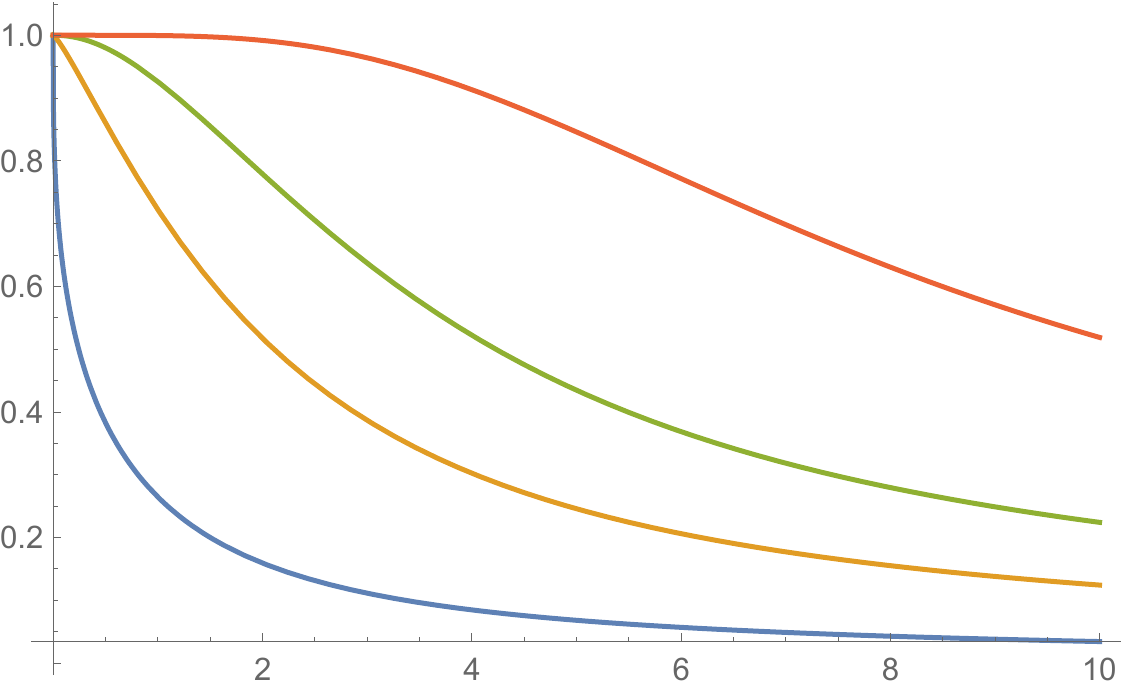}};
\node at (-2.41,-1.21) {\footnotesize $0.34$};
\node at (-1.7,-0.41) {\footnotesize $1.24$};
\node at (0.05,0.05) {\footnotesize $2.24$};
\node at (0.95,1.1) {\footnotesize $5.24$};
\end{tikzpicture} \hspace{1.5cm}
\begin{tikzpicture}[slave]
\node at (0,0) {\includegraphics[width=6.5cm]{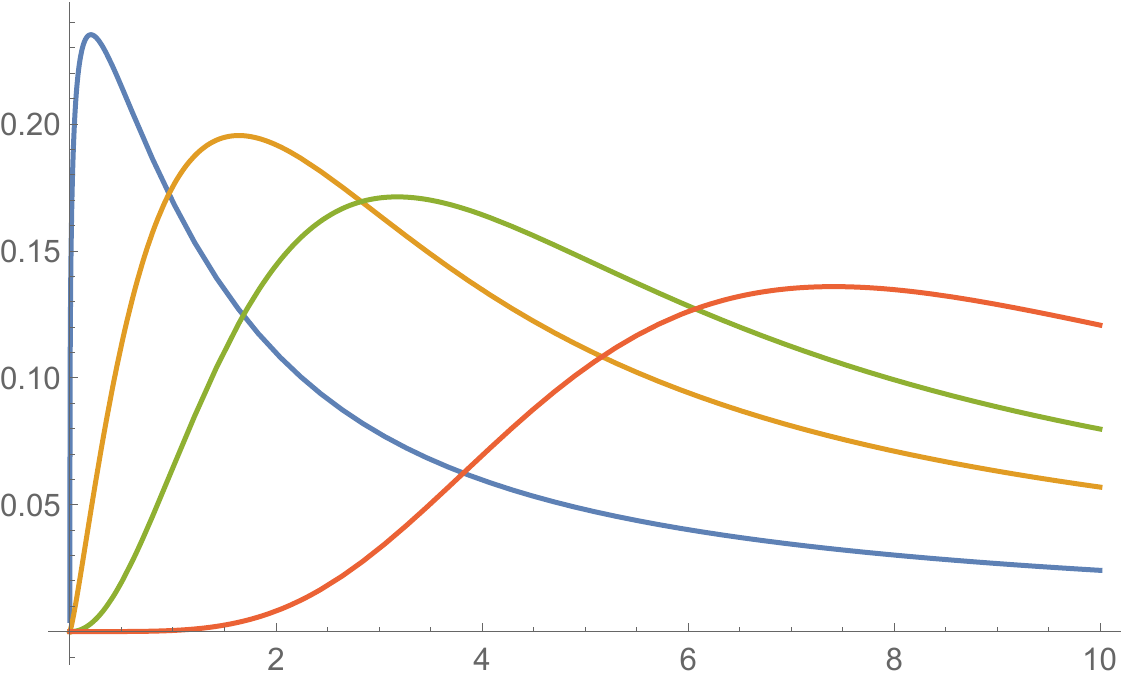}};
\node at (-2.5,2) {\footnotesize $0.34$};
\node at (-1.75,1.4) {\footnotesize $1.24$};
\node at (-0.3,0.9) {\footnotesize $2.24$};
\node at (2.5,0.4) {\footnotesize $5.24$};
\end{tikzpicture}
\end{center}
\vspace{-.5cm}
\caption{\label{fig: b11} The coefficients $t \mapsto b_{1}(t)$ (left) and $t \mapsto b_{(1,1)}(t,t)$ (right) for the indicated values of $\alpha$. }
\end{figure}
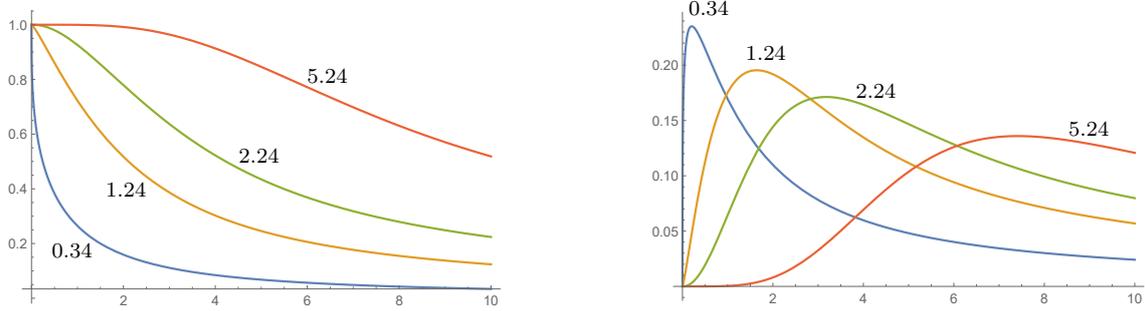
\begin{proof}[Proof of Corollary \ref{coro:correlation hard}]
Assertion (a) (except for the second equality in \eqref{def of b1}) is a direct consequence of \eqref{der of main result hard} and \eqref{joint cumulant}. The second equality in \eqref{def of b1} is obtained using \eqref{def of Q}, Fubini's theorem, and $\Gamma(\alpha+1)=\alpha \, \Gamma(\alpha)$:
\begin{align*}
& \int_{0}^{1} \mathrm{Q}(\alpha,t_{\ell}x)dx = \int_{0}^{1}dx \int_{t_{\ell}x}^{+\infty} dy \frac{e^{-y}y^{\alpha-1}}{\Gamma(\alpha)} = \bigg( \int_{0}^{t_{\ell}}dy \int_{0}^{\frac{y}{t_{\ell}}} dx + \int_{t_{\ell}}^{+\infty} dy \int_{0}^{1} dx \bigg)\frac{e^{-y}y^{\alpha-1}}{\Gamma(\alpha)} \\
& = \frac{\alpha}{t_{\ell}} \int_{0}^{t_{\ell}} \frac{e^{-y}y^{\alpha}}{\Gamma(\alpha+1)}dy + \int_{t_{\ell}}^{+\infty}\frac{e^{-y}y^{\alpha-1}}{\Gamma(\alpha)}dy = \alpha \frac{1-\mathrm{Q}(\alpha+1,t_{\ell})}{t_{\ell}} + \mathrm{Q}(\alpha,t_{\ell}).
\end{align*}
We now turn to the proof of (b). Using \eqref{asymp in main thm hard} with $u_\ell = \frac{i v_\ell }{\sqrt{b_{(1,1)}(t_\ell,t_\ell) n}}$ and $v_\ell \in \mathbb{R}$ fixed, we get
\begin{align*}
\mathbb{E}[e^{i \sum_{\ell = 1}^m v_\ell \mathcal{N}_\ell}]
& = \mathbb{E}[e^{\sum_{\ell = 1}^m u_\ell \mathrm{N}(r_{\ell})}]
e^{- \sum_{\ell = 1}^m u_\ell b_1(t_\ell) n}
	\\
& = e^{C_{1}(\vec{u}) n + C_{2}(\vec{u}) + \bigO(n^{-\frac{2\hat{\alpha}+\alpha}{2+\alpha}})}
e^{- \sum_{\ell = 1}^m u_\ell \partial_{u_\ell} C_1|_{\vec{u}=\vec{0}} n }
\end{align*}
as $n \to +\infty$, where the dependence of $C_{1}$ and $C_{2}$ in $\vec{u}$ has been made explicit. Since $C_j|_{\vec{u}=\vec{0}} = 0$ for $j = 1,2$ and $u_\ell = \bigO(n^{-1/2})$, we thus have
\begin{align*}
\mathbb{E}[&e^{i \sum_{\ell = 1}^m v_\ell \mathcal{N}_\ell}]
 =  e^{\frac{1}{2}\sum_{\ell,k= 1}^m u_\ell u_k \partial_{u_\ell}\partial_{u_k} C_1|_{\vec{u}=\vec{0}} n
+ \bigO(|\vec{u}|^3 n + |\vec{u}| + n^{-\frac{2\hat{\alpha}+\alpha}{2+\alpha}})}
	\\
& =  e^{\frac{1}{2}\sum_{\ell,k = 1}^m \frac{iv_\ell}{\sqrt{b_{(1,1)}(t_\ell,t_\ell)}} \frac{iv_k}{\sqrt{b_{(1,1)}(t_k,t_k)}} b_{(1,1)}(t_{\min(\ell,k)}, t_{\max(\ell,k)}) + o( 1)}
\to  e^{-\frac{1}{2}\sum_{\ell, k=1}^m v_\ell \Sigma_{\ell,k} v_k}
\end{align*}
as $n \to +\infty$. In other words, $\mathbb{E}[e^{i \sum_{\ell = 1}^m v_\ell \mathcal{N}_\ell}]$ converges pointwise to $e^{-\frac{1}{2}\sum_{\ell, k=1}^m v_\ell \Sigma_{\ell,k} v_k}$ as $n\to + \infty$, which implies Assertion (b) by L\'evy's continuity theorem.
\end{proof}

\textbf{Comparison with other works on counting statistics.} There has been a lot of interest recently on counting statistics of two-dimensional point processes. We will not attempt to survey this literature here, but refer the interested reader to \cite{BF2022} and the introduction of \cite{ABE}. Our main goal in this subsection is to compare Theorem \ref{thm:main thm hard} with the works \cite{ChLe2022, ACCL2022}. 

In \cite{ChLe2022}, the following Mittag-Leffler ensemble is considered:
\begin{align}\label{def of ML}
\frac{1}{n!\hat{Z}_{n}} \prod_{1 \leq j < k \leq n} |z_{k} -z_{j}|^{2} \prod_{j=1}^{n}|z_{j}|^{2a}e^{-n |z_{j}|^{2b}}d^{2}z_{j}, \qquad z_{1},\ldots,z_{n} \in \mathbb{C},
\end{align}
where $b>0$ and $a>-1$ are parameters of the model. The associated equilibrium measure $\mu^{\mathrm{ML}}$ is supported on the disk $\{|z| \leq b^{-\frac{1}{2b}}\}$ and given by $\mu^{\mathrm{ML}}(d^{2}z) = \frac{b^{2}}{\pi}|z|^{2b-2}d^{2}z$. Note that $\mu^{\mathrm{ML}}$ is absolutely continuous with respect to the Lebesgue measure $d^{2}z$ (this contrasts with the equilibrium measure $\mu$ of \eqref{def of point process hard}, which is purely singular). Let $m \in \mathbb{N}_{>0}$, $r \in (0,b^{-\frac{1}{2b}})$, $\mathfrak{s}_{1},\ldots,\mathfrak{s}_{m} \in \mathbb{R}$, $a > -1$ and $b>0$ be fixed parameters such that $\mathfrak{s}_{1}<\ldots<\mathfrak{s}_{m}$. The main result of \cite{ChLe2022} is the large $n$ asymptotics of the $m$-point moment generating function of the disk counting statistics of \eqref{def of ML} when the radii are merging either in the bulk, i.e. $r_{\ell} = r \big( 1+\frac{\sqrt{2}\, \mathfrak{s}_{\ell}}{r^{b}\sqrt{n}} \big)^{\frac{1}{2b}}$ for all $\ell\in\{1,\ldots,m\}$, or at the soft edge, i.e. $r_{\ell} = b^{-\frac{1}{2b}} \big( 1+\sqrt{2b}\frac{\mathfrak{s}_{\ell}}{\sqrt{n}} \big)^{\frac{1}{2b}}$ for all $\ell \in\{1,\ldots,m\}$. In both cases, it is shown in \cite{ChLe2022} that
\begin{align*}
\mathbb{E}^{\mathrm{ML}}\bigg[ \prod_{j=1}^{m} e^{u_{j}\mathrm{N}(r_{j})} \bigg] = \exp \bigg( D_{1} n + D_{2} \sqrt{n} + D_{3} +  \frac{D_{4}}{\sqrt{n}} + \bigO\bigg(\frac{(\ln n)^{2}}{n}\bigg)\bigg), \qquad \mbox{as } n \to + \infty,
\end{align*}
and the constants $D_{1},\ldots,D_{4}$ are determined explicitly. The constant $D_{1}$ is particularly simple; for example, in the bulk regime it is given by $D_{1} = \int_{|z|\leq r} \mu^{\mathrm{ML}}(d^{2}z) \sum_{j=1}^{m}u_{j} = b r^{2b} \sum_{j=1}^{m}u_{j}$. The constant $D_{2}$ is more complicated and given by
\begin{align*}
D_{2} = \begin{cases}
\ds \sqrt{2}\, b r^{b} \int_{-\infty}^{+\infty} \bigg( \ln \mathcal{H}^{\mathrm{ML}}(x; \vec{\mathfrak{s}},\vec{u}) - \chi_{(-\infty,0)}(x) \sum_{j=1}^{m}u_{j}\bigg) dx, & \mbox{for the bulk regime}, \\
\ds \sqrt{2b}\int_{-\infty}^{0} \bigg( \ln \mathcal{H}^{\mathrm{ML}}(x; \vec{\mathfrak{s}},\vec{u}) - \sum_{j=1}^{m}u_{j}\bigg) dx, & \mbox{for the soft edge regime},
\end{cases}
\end{align*}
where $\chi_{(-\infty,0)}(x)=1$ if $x<0$ and $\chi_{(-\infty,0)}(x)=0$ otherwise, $\vec{\mathfrak{s}} := (\mathfrak{s}_{1},\ldots,\mathfrak{s}_{m})$, and
\begin{align*}
\mathcal{H}^{\mathrm{ML}}(x; \vec{\mathfrak{s}},\vec{u}):= 1 + \sum_{\ell=1}^{m} (e^{u_{\ell}}-1)\exp\bigg[ \sum_{j=\ell+1}^{m}u_{j} \bigg] \frac{\mathrm{erfc}(x-\mathfrak{s}_{\ell})}{2}.
\end{align*}
We find it curious that the above function has the same structure as the function $\mathcal{H}_{\alpha}$ in \eqref{def of H}; namely, there are both of the form
\begin{align}\label{special form}
1 + \sum_{\ell=1}^{m} (e^{u_{\ell}}-1)\exp \bigg[ \sum_{j=\ell+1}^{m}u_{j} \bigg] X_{\ell}(x)
\end{align}
where $X_{\ell}(x):= \mathrm{Q}(\alpha,t_{\ell}x)$ in the present paper and $X_{\ell}(x):= \mathrm{erfc}(x-\mathfrak{s}_{\ell})/2$ in \cite{ChLe2022}. Note also that 
\begin{itemize}
\item $\mathcal{H}_{\alpha}$ already appears in the leading constant $C_{1}$, while $\mathcal{H}^{\mathrm{ML}}$ appears in $D_{2}$, 
\item $t_{1},\ldots,t_{m}$ are dilation parameters of $X_{\ell}$, in the sense that they appear in the multiplicative form ``$t_{\ell}x$" in $\mathrm{Q}(\alpha,t_{\ell}x)$, while $\mathfrak{s}_{1},\ldots,\mathfrak{s}_{m}$ are translation parameters of $X_{\ell}$, in the sense that they appear in the additive form ``$x-\mathfrak{s}_{\ell}$" in $\mathrm{erfc}(x-\mathfrak{s}_{\ell})/2$.
\end{itemize}
Let $0< \rho < b^{-\frac{1}{2b}}$ be fixed. The following point process was considered in \cite{ACCL2022}:
\begin{align}\label{ML hard}
\frac{1}{n!\tilde{Z}_{n}} \prod_{1 \leq j < k \leq n} |z_{k} -z_{j}|^{2} \prod_{j=1}^{n}|z_{j}|^{2\alpha}e^{-n |z_{j}|^{2b}}d^{2}z_{j}, \qquad |z_{j}|\leq \rho.
\end{align}
The only (but important) difference between the point processes \eqref{def of ML} and \eqref{ML hard} is that in \eqref{ML hard} the points are constrained to lie in the disk $\{|z| \leq \rho\}$. Because $\rho<b^{-\frac{1}{2b}}$, the circle $\{|z| = \rho\}$ is a hard wall of \eqref{ML hard} and it is shown in \cite{ACCL2022} that the associated equilibrium measure $\mu_{h}^{\mathrm{ML}}$ is given by
\begin{align}\label{eq mes}
& \mu_{h}^{\mathrm{ML}}(d^{2}z) = \mu_{\mathrm{reg}}^{\mathrm{ML}}(d^{2}z) + \mu_{\mathrm{sing}}^{\mathrm{ML}}(d^{2}z),\nonumber
\\ & \mu_{\mathrm{reg}}^{\mathrm{ML}}(d^{2}z) := 2b^{2}r^{2b-1}dr\frac{d\theta}{2\pi}, \quad \mu_{\mathrm{sing}}^{\mathrm{ML}}(d^{2}z) := c_{\rho} \delta_{\rho}(r) dr \frac{d\theta}{2\pi},
\end{align}
where $z=re^{i\theta}$, $r>0$, $\theta \in (-\pi,\pi]$ and $c_{\rho} := \int_{|z| > \rho}\mu^{\mathrm{ML}}(d^{2}z) = \int_{\rho}^{b^{-\frac{1}{2b}}} 2b^{2}r^{2b-1}dr = 1-b\rho^{2b}$. For the hard edge regime $r_{\ell} = \rho \big( 1-\frac{t_{\ell}}{n} \big)^{\frac{1}{2b}}$ with $t_{1}>\dots>t_{m}\geq 0$, it is proved in \cite{ACCL2022} that 
\begin{align}\label{hard edge asymptotic ML}
\mathbb{E}_{h}^{\mathrm{ML}}\bigg[ \prod_{j=1}^{m} e^{u_{j}\mathrm{N}(r_{j})} \bigg] = \exp \bigg(E_{1}n + E_{2}\ln n + E_{3} + \frac{E_{4}}{\sqrt{n}} + \bigO(n^{-\frac{3}{5}})\bigg), \qquad \mbox{as } n \to + \infty,
\end{align}
where $E_{1} = \int_{|z|\leq \rho} \mu_{\mathrm{reg}}^{\mathrm{ML}}(d^{2}z) \times \sum_{j=1}^{m}u_{j} + \int_{b\rho^{2b}}^{1} \ln \mathcal{H}_{h}^{\mathrm{ML}} (x;\vec{t},\vec{u})dx$ 
with
\begin{align*}
\mathcal{H}_{h}^{\mathrm{ML}} (x;\vec{t},\vec{u}) = 1 + \sum_{\ell=1}^{m} (e^{u_{\ell}}-1)\exp \bigg[ \sum_{j=\ell+1}^{m}u_{j} \bigg] e^{-\frac{t_{\ell}}{b}(x-b\rho^{2b})}.
\end{align*}
This function is also in the form \eqref{special form}, with $X_{\ell}(x) = e^{-\frac{t_{\ell}}{b}(x-b\rho^{2b})} = \mathrm{Q}(1,\frac{t_{\ell}}{b}(x-b\rho^{2b}))$. It is also interesting to note the presence of the term $E_{2}\ln n$ in \eqref{hard edge asymptotic ML}, while in \eqref{asymp in main thm hard} there is no term proportional to $\ln n$. We believe the reason for this is that $\mu_{h}^{\mathrm{ML}}$ has a non-trivial component $\mu_{\mathrm{reg}}^{\mathrm{ML}}$ which is absolutely continuous with respect to $d^{2}z$, while the equilibrium measure $\mu$ of \eqref{def of point process hard} is purely singular. This belief is supported by the following fact: when $\rho \to 0$, the measure $\mu_{h}^{\mathrm{ML}}$ becomes purely singular (because $c_{\rho}\to 1$), and $E_{2}\to 0$ (as can be easily checked from \cite[Theorem 1.3]{ACCL2022}).

\medskip A transition regime between the hard edge and the bulk was also considered in \cite{ACCL2022}. This regime is called ``the semi-hard edge regime" and corresponds to the case when the radii are at a distance of order $1/\sqrt{n}$ from the hard edge. More precisely, for $r_{\ell} = \rho \big( 1+\frac{\sqrt{2}\, \mathfrak{s}_{\ell}}{\rho^{b}\sqrt{n}} \big)^{\frac{1}{2b}}$ with $\mathfrak{s}_{1}<\dots<\mathfrak{s}_{m}<0$, we have
\begin{align*}
\mathbb{E}_{h}^{\mathrm{ML}}\bigg[ \prod_{j=1}^{m} e^{u_{j}\mathrm{N}(r_{j})} \bigg] = \exp \bigg(F_{1}n + F_{2}\sqrt{n} + F_{3} + \frac{F_{4}}{\sqrt{n}} + \bigO\bigg(\frac{(\ln n)^{4}}{n}\bigg)\bigg),
\end{align*}
where $F_{1} = \int_{|z|\leq \rho} \mu_{\mathrm{reg}}^{\mathrm{ML}}(d^{2}z) \sum_{j=1}^{m}u_{j} = b \rho^{2b} \sum_{j=1}^{m}u_{j}$ and
\begin{align*}
F_{2} = \sqrt{2} \, b \rho^{b} \int_{-\infty}^{+\infty} \Big( \ln \mathcal{H}^{\mathrm{ML}}_{\mathrm{sh}}(x; \vec{\mathfrak{s}},\vec{u})-\chi_{(-\infty,0)}(x) \sum_{j=1}^{m}u_{j} \Big)dx,
\end{align*}
with
\begin{align*}
\mathcal{H}^{\mathrm{ML}}_{\mathrm{sh}}(x; \vec{\mathfrak{s}},\vec{u}) = 1+\sum_{\ell=1}^{m} (e^{u_{\ell}}-1)\exp \bigg[ \sum_{j=\ell+1}^{m}u_{j} \bigg] \frac{\mathrm{erfc}(x-\mathfrak{s}_{\ell})}{\mathrm{erfc}(x)}.
\end{align*}
The function $\mathcal{H}^{\mathrm{ML}}_{\mathrm{sh}}$ is also in the form \eqref{special form}, with $X_{\ell}(x) = \frac{\mathrm{erfc}(x-\mathfrak{s}_{\ell})}{\mathrm{erfc}(x)}$. The above discussion is summarized in Figure \ref{fig:summary}. 
\begin{figure}[h!]
\begin{center}
\begin{tikzpicture}
\node at (0,0) {$
\begin{array}{|c|c|c|c|c|}\hline 
\mbox{Point process} & \mbox{Regime} & \ds \mathbb{E}\bigg[ \prod_{j=1}^{m} e^{u_{j}\mathrm{N}(r_{j})} \bigg] \quad \mbox{as } n \to +\infty & X_{\ell}(x) & \mbox{Ref}  \\ \hline 
\eqref{def of point process hard} & \mbox{Hard edge} &  \exp \big( C_{1} n + C_{2} + \bigO(n^{-\frac{2\hat{\alpha}+\alpha}{2+\alpha}}) \big)  & \mathrm{Q}(\alpha,t_{\ell}x) & \mbox{Theorem } \ref{thm:main thm hard} \rule{0pt}{0.45cm} \\[0.1cm] \hline 
\eqref{def of ML} & \begin{subarray}{l} \mbox{Bulk \& } \\[0.1cm] \mbox{Soft edge} \end{subarray}  &  \exp \big( D_{1} n + D_{2} \sqrt{n} + D_{3} +  \frac{D_{4}}{\sqrt{n}} + \bigO\big(\tfrac{(\ln n)^{2}}{n}\big)\big)  & \ds \frac{\mathrm{erfc}(x-\mathfrak{s}_{\ell})}{2} & \cite{ChLe2022} \rule{0pt}{0.6cm} \\[0.25cm] \hline 
\eqref{ML hard} & \mbox{Hard edge}  &  \exp \big(E_{1}n + E_{2}\ln n + E_{3} + \frac{E_{4}}{\sqrt{n}} + \bigO(n^{-\frac{3}{5}})\big) & e^{-\frac{t_{\ell}}{b}(x-b\rho^{2b})} & \cite{ACCL2022} \rule{0pt}{0.6cm} \\[0.25cm] \hline 
\eqref{ML hard} & \mbox{Semi-hard edge}  &  \exp \big(F_{1}n + F_{2}\sqrt{n} + F_{3} + \frac{F_{4}}{\sqrt{n}} + \bigO\big(\frac{(\ln n)^{4}}{n}\big)\big) & \ds \frac{\mathrm{erfc}(x-\mathfrak{s}_{\ell})}{\mathrm{erfc}(x)} & \cite{ACCL2022} \rule{0pt}{0.6cm} \\[0.25cm] \hline 
\end{array}
$};
\end{tikzpicture}
\end{center}
\vspace{-0.5cm}\caption{\label{fig:summary}Summary.}
\end{figure}

\section{Preliminaries}\label{section:prelim}
Let $\mathcal{E}_{n} := \mathbb{E}\big[ \prod_{\ell=1}^{m} e^{u_{\ell}\mathrm{N}(r_{\ell})} \big]$, and define
\begin{align}\label{def of w and omega}
w(z) = (1-|z|^{2})^{\alpha-1}\omega(|z|), \qquad \omega(x) := \prod_{\ell=1}^{m}\begin{cases}
e^{u_{\ell}}, & \mbox{if } x<r_{\ell}, \\
1, & \mbox{if } x \geq r_{\ell}.
\end{cases}
\end{align}
By rewriting $\prod_{1 \leq j < k \leq n} |z_{k} -z_{j}|^{2}$ as the product of two Vandermonde determinants, and then using standard algebraic manipulations, we get
\begin{align}
\mathcal{E}_{n} & = \frac{1}{n!Z_{n}} \int_{\mathbb{D}}\dots \int_{\mathbb{D}} \prod_{1 \leq j < k \leq n} |z_{k} -z_{j}|^{2} \prod_{j=1}^{n} w(z_{j}) d^{2}z_{j} \nonumber \\
& = \frac{1}{Z_{n}} \det \left( \int_{\D} z^{j} \overline{z}^{k} w(z) d^{2}z \right)_{j,k=0}^{n-1}. \label{def of Dn as n fold integral}
\end{align}
Since $w$ is rotation-invariant, only the diagonal elements in \eqref{def of Dn as n fold integral} are non-zero, and thus
\begin{align}
\mathcal{E}_{n} = \frac{1}{Z_{n}}(2\pi)^{n}\prod_{j=1}^{n}\int_{0}^{1}u^{2j-1}w(u)du. \label{simplified determinant}
\end{align}
Substituting \eqref{def of w and omega}, we then find
\begin{align}
 \mathcal{E}_{n} & = \prod_{j=1}^{n} \frac{\int_{0}^{1} x^{2j-1} (1-x^{2})^{\alpha-1} \omega(x) dx}{\int_{0}^{1} x^{2j-1} (1-x^{2})^{\alpha-1} dx}. \label{lnEn}
\end{align}
It will be convenient for us to rewrite $\omega$ (defined in \eqref{def of w and omega}) as follows:
\begin{align}\label{def of omegaell}
\omega(x) = \sum_{\ell=1}^{m+1}\omega_{\ell} \mathbf{1}_{[0,r_{\ell})}(x), \qquad \omega_{\ell} = \begin{cases}
e^{u_{\ell}+\dots+u_{m}}-e^{u_{\ell+1}+\dots+u_{m}}, & \mbox{if } \ell < m, \\
e^{u_{m}}-1, & \mbox{if } \ell=m, \\
1, & \mbox{if } \ell=m+1,
\end{cases}
\end{align}
where $r_{m+1}:=+\infty$. Using \eqref{def of omegaell} in \eqref{lnEn} yields the following expression for $\ln \mathcal{E}_n$:
\begin{align}
& \ln \mathcal{E}_{n} = \sum_{j=1}^{n} \ln \bigg(1  +\sum_{\ell=1}^{m} \omega_{\ell} F_{n,j,\ell} \bigg), \label{main exact formula} 
	\\ \label{def of Fnjell}
& F_{n,j,\ell} := \frac{\int_{0}^{r_{\ell}} x^{2j-1} (1-x^{2})^{\alpha-1} dx}{\int_{0}^{1} x^{2j-1} (1-x^{2})^{\alpha-1} dx} = \frac{\mathrm{B}(r_{\ell}^{2},j,\alpha)}{\mathrm{B}(j,\alpha)}, \qquad j=1,\ldots,n, \;\; \ell=1,\dots,m,
\end{align}
where $\mathrm{B}(j,\alpha)$ is the Beta function
\begin{align}\label{def of Beta}
\mathrm{B}(j,\alpha) := \int_{0}^{1} y^{j-1}(1-y)^{\alpha-1}dy = \frac{\Gamma(j)\Gamma(\alpha)}{\Gamma(j+\alpha)},
\end{align}
and $\mathrm{B}(v,j,\alpha)$ is the incomplete Beta function
\begin{align}\label{def of incomplete Beta}
\mathrm{B}(v,j,\alpha) := \int_{0}^{v} y^{j-1}(1-y)^{\alpha-1}dy.
\end{align}
Many properties of these functions are stated e.g. in \cite[Sections 5.12 and 8.17]{NIST}. It is also convenient for us to consider the normalized incomplete Beta function, which is given by
\begin{align}\label{def of I}
I(v,j,\alpha) := \frac{\mathrm{B}(v,j,\alpha)}{\mathrm{B}(j,\alpha)},
\end{align}
so that $F_{n,j,\ell} = I(r_{\ell}^{2},j,\alpha)$.

\medskip Hence, to analyze the right-hand side of \eqref{main exact formula}, we need the asymptotics of $I(v,j,\alpha)$ when $(v - 1)\asymp n^{-1}$ and simultaneously $j\in \{1,\ldots,n\}$ and $\alpha$ fixed. The large $n$ behavior of $I(r_{\ell}^{2},j,\alpha)$ depends crucially on whether $j$ remains bounded or not as $n\to + \infty$. We will therefore split the sum \eqref{main exact formula} in two parts as follows
\begin{align*}
\ln \mathcal{E}_{n} = S_{0} + S_{1}, \qquad \mbox{where} \quad S_{0} = \sum_{j=1}^{\frac{n}{M}-1} \ln \bigg(1  +\sum_{\ell=1}^{m} \omega_{\ell} F_{n,j,\ell} \bigg), \qquad S_{1} = \sum_{j=\frac{n}{M}}^{n} \ln \bigg( 1 + \sum_{\ell=1}^{m} \omega_{\ell} F_{n,j,\ell} \bigg),
\end{align*}
and $M:=n (\lceil\frac{n}{n^{\frac{2}{2+\alpha}}} \rceil)^{-1} \asymp n^{\frac{2}{2+\alpha}}$ is a new parameter such that $\N \ni n/M \to + \infty$ as $n\to + \infty$. (A more naive choice for $M$ would be $M = n/M'$ where $M'$ is large but fixed, but this choice does not yield a good control over certain error terms in the proof. The precise reason as to why we choose $M\asymp n^{\frac{2}{2+\alpha}}$ is technical and will become apparent at the end of Section \ref{section:proof edge}.) 

\medskip  The following lemma establishes an exact identity that will be useful to handle the sum $S_{0}$, i.e. to obtain the large $n$ asymptotics of $F_{n,j,\ell}$ when $j$ is ``not very large".
\begin{lemma}\label{lemma:Beta for j fixed}
Let $j\in \N_{>0}$, $\alpha>0$ and $v\in [0,1]$. Then we have the exact identity
\begin{align*}
I(v,j,\alpha) = 1 - \frac{(1-v)^{\alpha}}{\B(j,\alpha)} \sum_{p=0}^{j-1} (-1)^{p} \binom{j-1}{p} \frac{(1-v)^{p}}{\alpha+p}.
\end{align*}
\end{lemma}
\begin{proof}
The statement follows from \cite[eqs 8.17.4 and 8.17.7]{NIST}. We also provide a short proof here for convenience. Substituting $x^{j-1} = \sum_{p=0}^{j-1} (-1)^{p}\binom{j-1}{p} (1-x)^{p}$ in \eqref{def of I} yields
\begin{align*}
I(v,j,\alpha) & = \frac{1}{\B(j,\alpha)}\sum_{p=0}^{j-1} (-1)^{p} \binom{j-1}{p} \int_{0}^{v}(1-x)^{\alpha+p-1}dx \\
& = \frac{1}{\B(j,\alpha)} \bigg( \sum_{p=0}^{j-1} \binom{j-1}{p} \frac{(-1)^{p}}{\alpha+p} - (1-v)^{\alpha} \sum_{p=0}^{j-1} (-1)^{p} \binom{j-1}{p} \frac{(1-v)^{p}}{\alpha+p} \bigg).
\end{align*}
Replacing $v$ by $1$ above yields $1=\frac{1}{\B(j,\alpha)} \sum_{p=0}^{j-1} \binom{j-1}{p} \frac{(-1)^{p}}{\alpha+p}$, and the claim follows.
\end{proof}
To analyze $S_{1}$, we will use the uniform asymptotics of the incomplete Beta function (this is the main novelty of the proof, as earlier works such as \cite{ChLe2022, ACCL2022} on the Mittag-Leffler ensemble rely instead on the uniform asymptotics of the incomplete gamma function). The following lemma is due to Temme \cite[Section 11.3.3.1]{Temme} (this result can also be found in e.g. \cite[Section 8.18(ii)]{NIST}).
\begin{lemma}[Temme \cite{Temme}]\label{lemma:Temme}
Let $N\in \N_{>0}$. As $j\to +\infty$ with $\alpha >0$ fixed,
\begin{align*}
I(v,j,\alpha) = \frac{\Gamma(j+\alpha)}{\Gamma(j)}d_{0} F_{0} \bigg( 1 + \sum_{k=1}^{N-1} \frac{d_{k}F_{k}}{d_{0}F_{0}} + \bigO(j^{-N}) \bigg)
\end{align*}
uniformly for $v$ in compact subsets of $(0,1]$. The coefficients $F_{k}=F_{k}(v,j,\alpha)$ are defined by
\begin{align}\label{def Fk rec}
F_{k} = \frac{k-1+\alpha-j \ln(v^{-1})}{j}F_{k-1} + \frac{(k-1)\ln(v^{-1})}{j}F_{k-2}, \qquad k\geq 2,
\end{align}
with the initial assignments
\begin{align}\label{def FK initial}
F_{0} = j^{-\alpha}\mathrm{Q}(\alpha,j \ln(v^{-1})), \qquad F_{1} = \frac{\alpha-j \ln(v^{-1})}{j}F_{0} + \frac{(\ln(v^{-1}))^{\alpha}v^{j}}{j \Gamma(\alpha)},
\end{align}
where $\mathrm{Q}$ is defined in \eqref{def of Q} and the coefficients $d_{k}=d_{k}(v,\alpha)$ are defined through the generating function
\begin{align}\label{def of dk}
\bigg( \frac{1-e^{-t}}{t} \bigg)^{\alpha-1} = \sum_{k=0}^{\infty} d_{k}(t-\ln(v^{-1}))^{k}.
\end{align}
In particular,
\begin{align*}
d_{0} = \bigg( \frac{1-v}{\ln(v^{-1})} \bigg)^{\alpha-1}, \qquad d_{1} = \frac{(\alpha-1)(v-1+v\ln(v^{-1}))}{(v-1)^{2}} \bigg( \frac{1-v}{\ln(v^{-1})} \bigg)^{\alpha}.
\end{align*}
\end{lemma}

\begin{remark}(Determinants with circular root-type singularities.) 
Note from \eqref{def of Dn as n fold integral} that $\mathcal{E}_{n}$ can be seen as a ratio of two determinants. The determinant on the numerator involves $w$, and this weight has a root-type singularity along the unit circle (i.e. along the hard edge). Other determinants with circular root-type singularities have been considered in \cite{BC2022}; however, the singularities in \cite{BC2022} lie in the bulk, and the asymptotics of the corresponding determinants involve the so-called associated Hermite polynomials (this contrasts drastically with the asymptotics of $\mathcal{E}_{n}$, which are given in Theorem \ref{thm:main thm hard}).
\end{remark}

\begin{remark}(Partition function.)
Asymptotic expansions of partition functions of two-dimensional point processes are a classical topic of interest, see e.g. \cite[Section 5.3]{BF2022}. For rotation-invariant (and determinantal) ensembles with soft edges, precise formulas up to and including the term of order $1$ have been obtained in the recent work \cite{BKS2022}. The class of ensembles considered in  \cite{BKS2022} includes \eqref{def of point process hard} when $\alpha$ is proportional to $n$, see \cite[Section 4.2]{BKS2022}. As mentioned earlier, for $\alpha$ fixed, the ensemble \eqref{def of point process hard} has a hard edge and is therefore not considered in \cite{BKS2022}. As a minor aside, we compute here the partition function of \eqref{def of point process hard} with $\alpha$ fixed using a similar formula as \eqref{simplified determinant}. As in \eqref{simplified determinant} (but with $w(u)$ replaced by $(1-|u|^{2})^{\alpha-1}$), we get
\begin{align*}
Z_{n} & = (2\pi)^{n}\prod_{j=0}^{n-1}\int_{0}^{1}u^{2j+1}(1-u^{2})^{\alpha-1}du = \pi^{n} \prod_{j=1}^{n} \B(j,\alpha) = \pi^{n}\Gamma(\alpha)^{n} \frac{G(n+1)G(1+\alpha)}{G(n+1+\alpha)},
\end{align*}
where for the last identity we have used the functional equation for the Barnes $G$-function to write
\begin{equation}
\prod_{j=1}^{n} \Gamma (j+\alpha) = \frac{G(n+\alpha +1)}{G(1+\alpha)}.
\end{equation}
Using the expansion (see \cite[Eq. 5.17.5]{NIST})
\begin{equation}
\ln G(z+1) = \frac{z^{2}}{2}\ln z - \frac{3}{4}z^{2} + \frac{\ln(2\pi)}{2}z - \frac{1}{12}\ln z + \zeta'(-1) + \bigO(z^{-1}), \qquad z \to + \infty,
\end{equation}
we then get
\begin{align*}
Z_{n} = \exp \bigg( \hspace{-0.1cm} -\alpha \, n \ln n + \big( \alpha + \ln(\pi \Gamma(\alpha)\big)n - \frac{\alpha^{2}}{2}\ln n + \ln G(1+\alpha) - \frac{\alpha}{2} \ln(2\pi) + \bigO(n^{-1}) \bigg), \quad \mbox{as } n \to + \infty.
\end{align*}
\end{remark}

\section{Proof of Theorem \ref{thm:main thm hard}}\label{section:proof edge}

As mentioned in Section \ref{section:prelim}, it is convenient to split the sum \eqref{main exact formula} into two parts:
\begin{align}\label{log Dn as a sum of sums hard}
\ln \mathcal{E}_{n} = S_{0} + S_{1},
\end{align}
where
\begin{align}
& S_{0} = \sum_{j=1}^{\frac{n}{M}-1} \ln \bigg(1  +\sum_{\ell=1}^{m} \omega_{\ell} F_{n,j,\ell} \bigg), & & S_{1} = \sum_{j=\frac{n}{M}}^{n} \ln \bigg( 1 + \sum_{\ell=1}^{m} \omega_{\ell} F_{n,j,\ell} \bigg), \label{def of S0 and S1 hard}
\end{align}
and $M:=n (\lceil\frac{n}{n^{\frac{2}{2+\alpha}}} \rceil)^{-1} \asymp n^{\frac{2}{2+\alpha}}$. Define also $\Omega := e^{u_{1}+\dots+u_{m}}$. We first obtain the large $n$ asymptotics of $S_{0}$ using Lemma \ref{lemma:Beta for j fixed}.

\begin{lemma}\label{lemma: S0 hard}
Let $x_{1},\dots,x_{m} \in \mathbb{R}$ be fixed. There exists $\delta > 0$ such that
\begin{align*}
S_{0} = \Big(\frac{n}{M}-1\Big) \ln \Omega -  \frac{\sum_{\ell=1}^{m} \omega_{\ell}t_{\ell}}{\Omega \, \Gamma(2+\alpha)} \frac{n}{M^{1+\alpha}} - \frac{\sum_{\ell=1}^{m} \omega_{\ell}t_{\ell}}{\Omega \, \Gamma(1+\alpha)} \frac{\alpha-2}{2}\frac{1}{M^{\alpha}} + \bigO\bigg( \frac{n}{M^{2+\alpha}} + \frac{n}{M^{1+2\alpha}} + \frac{M}{n \, M^{\alpha}} \bigg), 
\end{align*}
as $n\to + \infty$, uniformly for $u_{1} \in \{z \in \mathbb{C}: |z-x_{1}|\leq \delta\},\dots,u_{m} \in \{z \in \mathbb{C}: |z-x_{m}|\leq \delta\}$.
\end{lemma}
\begin{proof}
Let $K:=n/M-1$. Recalling that $F_{n,j,\ell} = I(r_{\ell}^{2},j,\alpha)$ and $r_{\ell} = (1-\frac{t_{\ell}}{n})^{1/2}$, and using Lemma \ref{lemma:Beta for j fixed}, we infer that
\begin{align}\label{Fnjl expansion}
F_{n,j,\ell} = 1 - \frac{t_{\ell}^{\alpha}}{n^{\alpha}\B(j,\alpha)} \sum_{p=0}^{j-1} (-1)^{p} \binom{j-1}{p} \frac{t_{\ell}^{p}}{(\alpha+p)n^{p}} = 1-\frac{t_{\ell}^{\alpha}}{\alpha \, n^{\alpha}\B(j,\alpha)}+\bigO\bigg(\frac{n^{-\alpha-1}j}{\B(j,\alpha)}\bigg), 
\end{align}
as $n \to + \infty$ uniformly for $j \in \{1,\dots,K\}$ and $\ell\in\{1,\dots,m\}$. Using (see e.g. \cite[formula 5.11.1]{NIST})
\begin{align}\label{asymp gamma}
\ln \Gamma(z) = (z-\tfrac{1}{2})\ln z - z + \tfrac{1}{2} \ln(2\pi) + \bigO(z^{-1}), \qquad \mbox{as } z \to + \infty,
\end{align}
we infer that $\B(j,\alpha) = \bigO(j^{-\alpha})$ as $j\to +\infty$ with $\alpha$ fixed, so that the error term in \eqref{Fnjl expansion} can be replaced by $\bigO((j/n)^{\alpha+1})$. Hence, since $1+\sum_{\ell=1}^{m} \omega_{\ell} = e^{u_{1}+\dots+u_{m}} = \Omega$,
\begin{align*}
S_{0} & = \sum_{j=1}^{K} \ln \bigg( 1 + \sum_{\ell=1}^{m} \omega_{\ell} \bigg[1 - \frac{t_{\ell}^{\alpha}}{\alpha \, n^{\alpha}\B(j,\alpha)} + \bigO\big((\tfrac{j}{n})^{1+\alpha}\big) \bigg] \bigg) = \sum_{j=1}^{K} \ln \bigg( \Omega - \sum_{\ell=1}^{m}  \frac{\omega_{\ell} t_{\ell}^{\alpha}}{\alpha \, n^{\alpha}\B(j,\alpha)} + \bigO\big((\tfrac{j}{n})^{1+\alpha}\big)  \bigg).
\end{align*}
The $\bigO$-term after the first equality is clearly independent of $u_{1},\dots,u_{m}$, and therefore the $\bigO$-term after the second equality is uniform for $u_{1} \in \{z \in \mathbb{C}: |z-x_{1}|\leq \delta\},\dots,u_{m} \in \{z \in \mathbb{C}: |z-x_{m}|\leq \delta\}$, for any fixed $\delta>0$. We can (and do) choose $\delta>0$ sufficiently small such that $\Omega$ remains bounded away from $(-\infty,0]$ for $u_{1} \in \{z \in \mathbb{C}: |z-x_{1}|\leq \delta\},\dots,u_{m} \in \{z \in \mathbb{C}: |z-x_{m}|\leq \delta\}$, so that 
\begin{align*}
S_{0} & = \sum_{j=1}^{K} \bigg( \ln \Omega - \frac{\sum_{\ell=1}^{m} \omega_{\ell}t_{\ell}}{\alpha \, \Omega \, \B(j,\alpha) \; n^{\alpha}} + \bigO\bigg( \Big(\frac{j}{n}\Big)^{1+\alpha} + \Big(\frac{j}{n}\Big)^{2\alpha} \bigg) \bigg).
\end{align*}
 Furthermore,
\begin{align*}
\sum_{j=1}^{K} \bigg( \Big(\frac{j}{n}\Big)^{1+\alpha} + \Big(\frac{j}{n}\Big)^{2\alpha} \bigg) = \bigO\bigg( K\Big(\frac{K}{n}\Big)^{1+\alpha} + K\Big(\frac{K}{n}\Big)^{2\alpha} \bigg), \qquad \mbox{as } n \to + \infty.
\end{align*}
The above $\bigO$-term can also be written as $\bigO(\frac{n}{M^{2+\alpha}} + \frac{n}{M^{1+2\alpha}})$, and thus
\begin{align}\label{asymp S0 in proof}
S_{0} = \Big(\frac{n}{M}-1\Big) \ln \Omega - \frac{\sum_{\ell=1}^{m} \omega_{\ell}t_{\ell}}{\alpha \, \Omega \; n^{\alpha}} \sum_{j=1}^{K}\frac{1}{\B(j,\alpha)} + \bigO\bigg( \frac{n}{M^{2+\alpha}} + \frac{n}{M^{1+2\alpha}} \bigg), \qquad \mbox{as } n \to + \infty.
\end{align}
By \eqref{def of Beta} and the functional relation $\Gamma(z+1)=z\Gamma(z)$,
\begin{align}\label{sum reciprocal of Beta}
\sum_{j=1}^{K}\frac{1}{\B(j,\alpha)} = \sum_{j=1}^{K}\frac{\Gamma(j+\alpha)}{\Gamma(j)\Gamma(\alpha)} = \frac{K \, \Gamma(1+K+\alpha)}{\Gamma(1+K)\Gamma(\alpha)} \frac{1}{K+\alpha} \sum_{j=1}^{K} \prod_{k=1}^{j-1} \frac{K-k}{K-k+\alpha}.
\end{align}
This last sum can be rewritten as
\begin{align*}
\sum_{j=1}^{K} \prod_{k=1}^{j-1} \frac{K-k}{K-k+\alpha} = \sum_{j=1}^{K} \frac{(K-1)!}{(K-j)!}\frac{\Gamma(K+\alpha+1-j)}{\Gamma(K+\alpha)} = \sum_{j=1}^{K}  \frac{\binom{K-1}{j-1}}{\binom{K-1+\alpha}{j-1}}.
\end{align*}
Using $\binom{K}{\ell}\binom{\ell}{k}=\binom{K}{k}\binom{K-k}{\ell-k}$ and then changing indices of summations, we obtain
\begin{align*}
\sum_{j=1}^{K} \prod_{k=1}^{j-1} \frac{K-k}{K-k+\alpha} = \frac{1}{\binom{K-1+\alpha}{K-1}}\sum_{j=0}^{K-1}\binom{K-1+\alpha-j}{K-1-j} = \frac{1}{\binom{K-1+\alpha}{K-1}}\sum_{j=0}^{K-1}\binom{j+\alpha}{j}.
\end{align*}
We now use the so-called ``parallel summation formula" $\sum_{j=0}^{K-1}\binom{j+\alpha}{j} = \binom{K+\alpha}{K-1}$ to find
\begin{align*}
\frac{1}{K+\alpha} \sum_{j=1}^{K} \prod_{k=1}^{j-1} \frac{K-k}{K-k+\alpha} = \frac{1}{1+\alpha}.
\end{align*}
Substituting the above in \eqref{sum reciprocal of Beta} yields
\begin{align*}
\sum_{j=1}^{K}\frac{1}{\B(j,\alpha)} = \frac{K \, \Gamma(1+K+\alpha)}{\Gamma(1+K)\Gamma(\alpha)} \frac{1}{1+\alpha}.
\end{align*}
This formula can easily be expanded as $K=\frac{n}{M}-1\to +\infty$ using \eqref{asymp gamma}:
\begin{align}\label{expansion of reciprocal Beta}
\frac{1}{n^{\alpha}}\sum_{j=1}^{K}\frac{1}{\B(j,\alpha)} = \frac{1}{M^{\alpha}}\frac{1}{(1+\alpha)\Gamma(\alpha)} \bigg\{ \frac{n}{M} + \frac{(1+\alpha)(\alpha-2)}{2} + \bigO\Big( \frac{M}{n} \Big) \bigg\}.
\end{align}
Substituting \eqref{expansion of reciprocal Beta} in \eqref{asymp S0 in proof} yields the claim.
\end{proof}
We now turn to the analysis of $S_{1}$. We will rely on Lemma \ref{lemma:Temme}, as well as on the following Riemann sum approximation lemma (whose proof is omitted).
\begin{lemma}\label{lemma:Riemann sum NEW}
Let $A=A(n)$, $B=B(n)$ be bounded functions of $n \in \{1,2,\dots\}$, such that
\begin{align*}
& a_{n} := An \qquad \mbox{ and } \qquad b_{n} := Bn
\end{align*}
are integers. Assume also that $B-A$ is positive and remains bounded away from $0$ as $n\to + \infty$. Let $f$ be a function independent of $n$, which is $C^{2}([A,B])$ for all $n\in \{1,2,\dots\}$. Then as $n \to + \infty$, we have
\begin{align}
 \sum_{j=a_{n}}^{b_{n}}f(\tfrac{j}{n}) = n \int_{A}^{B}f(x)dx + \frac{f(A)+f(B)}{2} + \bigO \bigg( \frac{f'(A)+f'(B)}{n} + \sum_{j=a_{n}}^{b_{n}-1} \frac{\mathfrak{m}_{j,n}(f'')}{n^{2}} \bigg), \label{sum f asymp gap NEW}
\end{align}
where, for a given function $g$ continuous on $[A,B]$ and $j \in \{a_{n},\dots,b_{n}-1\}$, $\mathfrak{m}_{j,n}(g) := \max_{x \in [\frac{j}{n},\frac{j+1}{n}]}|g(x)|$.
\end{lemma}
Let $\mathsf{T}_{0} := \sum_{\ell=1}^{m} \omega_{\ell}t_{\ell}^{\alpha}$.
\begin{lemma}\label{lemma: S2km1 hard}
For any fixed $x_{1},\dots,x_{m} \in \mathbb{R}$, there exists $\delta > 0$ such that 
\begin{align*}
& S_{1} = n \bigg\{ \int_{0}^{1}f_{0}(x)dx - \frac{1}{M} \ln \Omega + \frac{\mathsf{T}_{0}}{\Omega \, \Gamma(\alpha + 2)} \frac{1}{M^{\alpha+1}} \bigg\} + \int_{0}^{1} f_{1}(x)dx + \frac{f_{0}(1) + \ln \Omega}{2} \\
& \quad - \frac{(2-\alpha)\mathsf{T}_{0}}{2\Omega \, \Gamma(\alpha+1)} \frac{1}{M^{\alpha}} + \bigO\bigg( \frac{n}{M^{1+2\alpha}}+\frac{n}{M^{2+\alpha}} + \frac{1+M^{1-\alpha}}{n} \bigg),
\end{align*}
as $n \to +\infty$ uniformly for $u_{1} \in \{z \in \mathbb{C}: |z-x_{1}|\leq \delta\},\dots,u_{m} \in \{z \in \mathbb{C}: |z-x_{m}|\leq \delta\}$, where
\begin{align}
& f_{0}(x) := \ln (1+\Sigma_{0}(x)), & & f_{1}(x) := \frac{\Sigma_{1}(x)}{1+\Sigma_{0}(x)}, \label{def of f0 and f1} \\
& \Sigma_{0}(x) := \sum_{\ell=1}^{m} \omega_{\ell} \mathrm{Q}(\alpha,t_{\ell} x), & & \Sigma_{1}(x) := \frac{x^{\alpha}}{x} \sum_{\ell=1}^{m} \omega_{\ell} \frac{e^{-t_{\ell}x}t_{\ell}^{\alpha}}{\Gamma(\alpha)} \frac{1-\alpha-t_{\ell}x}{2}. \label{def of Sigma0 and Sigma1}
\end{align}
\end{lemma}
\begin{proof}
By Lemma \ref{lemma:Temme}, since $n/M\to+\infty$, we have 
\begin{align*}
F_{n,j,\ell} = \frac{\Gamma(j+\alpha)}{\Gamma(j)}d_{0}(r_{\ell}^{2},\alpha)F_{0}(r_{\ell}^{2},j,\alpha) \bigg( 1 + \sum_{k=1}^{3} \frac{d_{k}(r_{\ell}^{2},\alpha)F_{k}(r_{\ell}^{2},j,\alpha)}{d_{0}(r_{\ell}^{2},\alpha)F_{0}(r_{\ell}^{2},j,\alpha)} + \bigO(j^{-4}) \bigg)
\end{align*}
as $n\to + \infty$ uniformly for $\ell \in \{1,\ldots,m\}$ and $j \in \{\frac{n}{M},\ldots,n\}$. Moreover, using \eqref{def Fk rec}, \eqref{def FK initial}, \eqref{def of dk} and \eqref{asymp gamma}, we get
\begin{align*}
F_{n,j,\ell} = \mathrm{Q}(\alpha, t_{\ell} j/n) + \frac{(j/n)^{\alpha}}{j} \frac{e^{-t_{\ell}j/n}t_{\ell}^{\alpha}}{\Gamma(\alpha)} \frac{1-\alpha - t_{\ell}j/n}{2} + \bigO\Big(\frac{1}{n^{2}}\frac{(j/n)^{\alpha}}{(j/n)^{2}}\Big), \qquad \mbox{as } n \to + \infty,
\end{align*}
uniformly for $\ell \in \{1,\ldots,m\}$ and $j \in \{\frac{n}{M},\ldots,n\}$. Hence, for small enough $\delta>0$,
\begin{align}\label{ln 1+F asymp}
\ln \bigg( 1+\sum_{\ell=1}^{m} \omega_{\ell} F_{n,j,\ell} \bigg) = f_{0}(j/n) + \frac{f_{1}(j/n)}{n}  + \bigO\Big(\frac{1}{n^{2}}\frac{(j/n)^{\alpha}}{(j/n)^{2}}\Big), \qquad \mbox{as } n \to + \infty,
\end{align}
uniformly for $j \in \{\frac{n}{M},\ldots,n\}$ and $u_{1} \in \{z \in \mathbb{C}: |z-x_{1}|\leq \delta\},\dots,u_{m} \in \{z \in \mathbb{C}: |z-x_{m}|\leq \delta\}$. Note that
\begin{align*}
\frac{1}{n^{2}} \sum_{j=\frac{n}{M}}^{n} \frac{(j/n)^{\alpha}}{(j/n)^{2}} = \bigO\bigg( \frac{1+M^{1-\alpha}}{n} \bigg).
\end{align*}
Furthermore, by Lemma \ref{lemma:Riemann sum NEW} with $A=\frac{1}{M}$, $B=1$,
\begin{align}
& \sum_{j=\frac{n}{M}}^{n} f_{0}(j/n) = n \int_{M^{-1}}^{1}f_{0}(x)dx + \frac{f_{0}(M^{-1})+f_{0}(1)}{2} + \bigO\bigg( \frac{M^{1-\alpha}+1}{n} \bigg), \label{sum f0} \\
& \frac{1}{n} \sum_{j=\frac{n}{M}}^{n} f_{1}(j/n) = \int_{M^{-1}}^{1}f_{1}(x)dx + \bigO\bigg( \frac{M^{1-\alpha}+1}{n} \bigg), \label{sum f1}
\end{align}
where, to estimate the $\bigO$-terms, we have used $n^{-1}f_{0}'(A) \lesssim n^{-1}A^{\alpha-1} = \bigO(\frac{M^{1-\alpha}}{n})$, $n^{-1}f_{0}'(B)=\bigO(n^{-1})$, $n^{-1}f_{1}(A) \lesssim n^{-1}A^{\alpha-1} = \bigO(\frac{M^{1-\alpha}}{n})$, and $n^{-1}f_{1}(B)=\bigO(n^{-1})$. Also, using \eqref{def of f0 and f1}--\eqref{def of Sigma0 and Sigma1}, as $n\to \infty$ we get
\begin{align*}
& n \int_{M^{-1}}^{1}f_{0}(x)dx = n \int_{0}^{1}f_{0}(x)dx - \frac{n}{M} \ln \Omega + \frac{\mathsf{T}_{0}}{\Omega \, \Gamma(\alpha+2)}\frac{n}{M^{1+\alpha}} + \bigO\Big( \frac{n}{M^{1+2\alpha}}+\frac{n}{M^{2+\alpha}} \Big), \\
& f_{0}(M^{-1}) = \ln \Omega - \frac{\mathsf{T}_{0}}{\Omega \, \Gamma(\alpha+1)}\frac{1}{M^{\alpha}} + \bigO\Big( \frac{1}{M^{1+\alpha}}+\frac{1}{M^{2\alpha}} \Big), \\
& \int_{M^{-1}}^{1}f_{1}(x)dx = \int_{0}^{1}f_{1}(x)dx - \frac{\mathsf{T}_{0}(1-\alpha)}{2\Omega \, \Gamma(\alpha+1)} \frac{1}{M^{\alpha}} + \bigO\Big( \frac{1}{M^{1+\alpha}}+\frac{1}{M^{2\alpha}} \Big).
\end{align*}
Substituting the above in \eqref{sum f0}--\eqref{sum f1} and then in \eqref{ln 1+F asymp} yields
\begin{align*}
S_{1} & = \sum_{j=\frac{n}{M}}^{n} f_{0}(j/n) + \sum_{j=\frac{n}{M}}^{n} \frac{f_{1}(j/n)}{n}  + \bigO\bigg( \frac{1+M^{1-\alpha}}{n} \bigg) \\
& = n \bigg\{ \int_{0}^{1}f_{0}(x)dx - \frac{1}{M} \ln \Omega + \frac{\mathsf{T}_{0}}{\Omega \, \Gamma(\alpha + 2)} \frac{1}{M^{\alpha+1}} \bigg\} + \int_{0}^{1} f_{1}(x)dx + \frac{f_{0}(1) + \ln \Omega}{2} \\
& \quad - \frac{(2-\alpha)\mathsf{T}_{0}}{2\Omega \, \Gamma(\alpha+1)} \frac{1}{M^{\alpha}} + \bigO\bigg( \frac{n}{M^{1+2\alpha}}+\frac{n}{M^{2+\alpha}} + \frac{1+M^{1-\alpha}}{n} \bigg)
\end{align*}
\end{proof}

\begin{proof}[Proof of Theorem \ref{thm:main thm hard}]
Combining Lemmas \ref{lemma: S0 hard} and \ref{lemma: S2km1 hard} yields
\begin{align*}
\ln \mathcal{E}_{n} = n \int_{0}^{1}f_{0}(x)dx + \int_{0}^{1}f_{1}(x)dx + \frac{f_{0}(1) - \ln \Omega}{2} + \bigO\bigg( \frac{n}{M^{1+2\alpha}}+\frac{n}{M^{2+\alpha}} + \frac{1+M^{1-\alpha}}{n} \bigg).
\end{align*}
The above $\bigO$-term can be rewritten as
\begin{align*}
\bigO\bigg( \frac{n}{M^{1+\alpha+\hat{\alpha}}} + \frac{M^{1-\hat{\alpha}}}{n} \bigg),
\end{align*}
where $\hat{\alpha} = \min\{\alpha,1\}$. Since $M \asymp n^{\frac{2}{2+\alpha}}$, we have 
\begin{align*}
\frac{n}{M^{1+\alpha+\hat{\alpha}}} \asymp  \frac{M^{1-\hat{\alpha}}}{n} \asymp n^{-\frac{2\hat{\alpha}+\alpha}{2+\alpha}},
\end{align*}
which finishes the proof of Theorem \ref{thm:main thm hard}.
\end{proof}

\paragraph{Acknowledgements.} CC acknowledges support from the Swedish Research Council, Grant No. 2021-04626. PM acknowledges support from the Magnusons fond, Grant No. MG2022-0014, and from the European Research Council, Grant No. 715539.

\end{document}